\documentclass[10pt]{article}

\usepackage{amsmath, latexsym, amssymb, amscd, amsfonts, amsthm, cite, epsfig, hyperref, enumerate, graphics, caption, subcaption}

\usepackage{epstopdf}

\usepackage{todonotes}

\usepackage[UKenglish]{babel}

\usepackage[T1]{fontenc} 
\usepackage{eulervm}   
\usepackage{tgpagella} 

\usepackage[all]{xy}
\graphicspath{{figures/}}

\newcommand{\Tr}{\text{Tr}}							
\newcommand{\id}{\text{d}}						
\newcommand{\set}[1]{ \left\lbrace #1 \right\rbrace }	

\newcommand{\bra}[1]{\langle #1 \vert}
\newcommand{\ket}[1]{\vert #1 \rangle}

\newcommand{\be}{\begin{equation}}
\newcommand{\ee}{\end{equation}}
\newcommand{\inpr}[1]{\left\langle #1 \right\rangle}
\newcommand{\norm}[1]{\lVert #1 \rVert}
\newcommand{\abs}[1]{\lvert #1 \rvert}

\newcommand{\rss}{\rho_{{\rm ss}}}

\begin{document}


\theoremstyle{definition}
\newtheorem{thm}{Theorem}
\newtheorem{defn}[thm]{Definition}
\newtheorem{lem}[thm]{Lemma}
\newtheorem{prop}[thm]{Proposition}

\newtheorem{col}[thm]{Corollary}

\title{Large Deviations, Central Limit and dynamical phase transitions in the atom maser}

\author{Federico Girotti$^1$, Merlijn van Horssen$^2$, Raffaella Carbone$^1$, M{\u a}d{\u a}lin Gu{\c t}{\u a}$^{2,}$\footnote{madalin.guta@nottingham.ac.uk}\\[4mm]
$^1$ 
Dipartimento di Matematica dell'Universit\`a di Pavia\\Via Ferrata 1, 27100 Pavia, Italy\\[2mm]
$^2$ University of Nottingham, School of Mathematical Sciences\\
University Park, NG7 2RD Nottingham, UK}

\maketitle

\begin{abstract}
The theory of quantum jump trajectories provides a new framework for understanding dynamical phase transitions in open systems. A candidate for such transitions is the  atom maser, which for certain parameters exhibits strong intermittency in the atom detection counts, and has a bistable  stationary state. Although previous numerical results suggested that the "free energy" may not be a smooth function, we show that  the atom detection counts satisfy a large deviations principle, and therefore we deal with a phase cross-over rather than a genuine phase transition. We argue however that the latter occurs in the limit of infinite pumping rate. As a corollary, we obtain the Central Limit Theorem for the counting process. The proof relies on the analysis of a certain deformed generator whose spectral bound is the limiting cumulant generating function. The latter is shown to be smooth, so that a large deviations principle holds by the G\"{a}rtner-Ellis Theorem. One of the main ingredients is the Krein-Rutman theory which extends the Perron-Frobenius theorem to a general class of positive compact semigroups. 

\end{abstract}

\newpage

\section{Introduction}


The last couple of decades have witnessed a revolution in the experimental realisation of quantum systems \cite{Haroche&Raimond}. Ultracold atomic gases are created and used for the study of complex many body phenomena such as quantum phase transitions \cite{Sachdev} shedding light on open problems in condensed-matter physics \cite{Bloch&Dalibard}.

Real quantum systems are "open" in the sense that they interact with their environment, which leads to an irreversible loss of coherence and to energy dissipation. In many cases, the dynamics can be well described by the Markov approximation in which the environment possesses no memory and interacts weakly with the system. The joint unitary evolution of the system and environment can be described through the input-output formalism \cite{Gardiner2004} using quantum stochastic calculus \cite{Parthasarathy1992}. In this framework, the Markov semigroup can be seen as the result of averaging over  stochastic quantum trajectories arising from continuous-time measurements performed in the environment. These are described via stochastic Schr\"{o}dinger (or filtering) equations \cite{Belavkin,Plenio1998} and capture the system's evolution conditional on the detection record. 


In  \cite{Garrahan2009} a new perspective was put forward, which looks at quantum jumps from the viewpoint of non-equilibrium statistical mechanics \cite{Evans&Franz}. 
Detection trajectories are seen as "configurations" of a stochastic system, and  large deviations theory \cite{Dembo2010,Ellis1995} is employed to study the \emph{dynamical phase transitions} arising in this way. Consider for simplicity the case of a counting measurement, which is directly relevant for the model studied in this paper. 
The interesting scenarios are that of a \emph{phase cross-over} in which the counting trajectories show intermittency between long active periods (many counts) and passive ones (few counts), and that of \emph{phase coexistence} where the counting process exhibit a mixture of infinitely long trajectories of either type. In the latter case, the asymptotic cumulant generating function (or "free energy") of the total counts process $\Lambda_t$ is singular at the origin, and the total counts do not obey a large deviations principle (LDP). In contrast, in a phase cross-over an LDP may hold but numerically and practically there would be a strong resemblance to an actual phase transition. 

For finite dimensional systems the counting process $\Lambda_t$ satisfies an LDP when the Markov dynamics is mixing, i.e. irreducible and aperiodic \cite{Hiai2007}. The proof uses the G\"{a}rtner-Ellis theorem according to which it suffices to prove the convergence of the cumulant generating function to a smooth limit. By the Markov property, the former can be expressed in terms 
of a certain "deformed generator" $\mathcal{L}_s$, and the existence of the limit
$$
\lim_{t\to\infty}\frac{1}{t} \log \mathbb{E}_\rho(e^{s\Lambda_t})  = \lim_{t\to\infty}\frac{1}{t} \log {\rm Tr}(\rho e^{t\mathcal{L}_s} (\mathbf{1}))=\lambda(s)
$$
follows from the spectral gap property of $\mathcal{L}_s$, where $\lambda(s)$ is the spectral bound of $\mathcal{L}_s$, and $\rho$ is the initial state.

In this paper we investigate the existence of dynamical phase transitions for the \emph{atom maser}, a well known quantum open system exhibiting interesting properties such as bistability and sub-Poissonian statistics \cite{Briegel1994,Englert1998,Walther&Varcoe}. The maser consists of a beam of excited atoms passing through a cavity with which they interact according to the Jaynes-Cummings model. After the interaction the atoms are measured in the standard basis and the trajectory of measurement outcomes is recorded. For certain values of the interaction strength, the stationary mean photon number changes abruptly (cf. Figure \ref{fig:stationary}), and the distribution is bistable, having a low and a high energy "phase". The measurement trajectories alternate between periods of low and high ground state atoms counts (cf. Figure \ref{fig:trajectories}), and its limiting moment generating function exhibits characteristic phase separation lines (cf. Figure \ref{fig:grids}).

Our main result (Theorem \ref{th.main}) is that the counts process satisfies the LDP, and therefore the atom maser does not have the non-analytic properties characteristic of phase transitions,  although it exhibits clear phase cross-over(s) which become sharper with increasing pumping rate. As a corollary, we obtain the Central Limit Theorem for the counting process, using a result of \cite{Bryc}. The proof follows the line of  \cite{Hiai2007}, but the novelty here is the treatment of an infinite dimensional system in continuous time dynamics. We use an $L^2$-representation \cite{Fagnola1994,Carbone2000} of the semigroup generated by $\mathcal{L}_s$ and show that the corresponding semigroup is compact. We then use the Krein-Rutman theory ( \cite{Chang2020} and references therein) to establish the uniqueness and strict positivity of the eigevector of $\lambda(s)$, and hence the existence of the spectral gap. Some steps of the proof rely on a special feature of the maser dynamics which allows us to restrict the attention to the commutative invariant algebra of diagonal operators. However the line of the proof is applicable to general infinite dimensional quantum Markov dynamics.


 For recent work on quantum dynamical phase transitions we refer to \cite{Garrahan2009,Garrahan2007,Hedges2009, Garrahan2008,Budini2010}. In particular, our investigation was motivated  by the numerical results of \cite{Garrahan2011} indicating a possible non-analytic behaviour of $\lambda(s)$. In \cite{Hiai2007} (see also \cite{Ogata2010}) a large deviation principle is shown to hold for correlated states on quantum spin chains; large deviations for quantum Markov semigroups are studied in \cite{Comman2008}. Metastable behaviour in a different atom maser has been investigated in \cite{Bruneau2008}. More broadly, there is a large body of large deviations work in quantum systems \cite{Hiai&Petz,Ogawa&Nagaoka,Lebowitz&Spohn,Lenci&Rey-Bellet,Derezinski&DeRoeck,Bjelakovic}.

In Section 
\ref{sec.background} we introduce the background of our problem: the atom 
maser and its Markov semigroup, the counting processes associated to the jump terms in the Lindblad generator, the static and dynamical phase transitions and the interplay between them, and the general setup of large deviations theory. In addition, the existence and properties of various semigroups are established rigorously. In Section \ref{sec.main}
we formulate the large deviations results and give a point by point outline of the proof. 
The results of a detailed numerical analysis are presented in Section \ref{sec.numerics}, where we argue that  "phase transitions" do occur in the limit of very large pumping rate, at $\alpha\approx 1$ (second order), at $\alpha\approx 6.66$ and further points (first order), where $\alpha$ is the pumping parameter (see Fig. \ref{fig:stationary} below).

\section{Background}
\label{sec.background}
In this section we introduce the atom maser dynamics, investigate the counting process associated to the measurement of outgoing atoms, and describe the basic elements of large deviations theory used in the paper. Propositions \ref{prop:def} and \ref{prop:defdef} establish the mathematical properties of the quantum dynamical semigroups used in the paper.

\subsection{The atom maser}

In the atom maser, two-level atoms pass successively through a cavity and interact resonantly with the electromagnetic field inside the cavity. The two-level atoms are identically and independently prepared  in the excited state, and for simplicity we assume that only a single atom passes through the cavity at any time. 
In addition, the cavity is also coupled to a thermal bath which represents the interaction between the (non-ideal) cavity and the environment. The combined effects of the interactions with the atoms and the environment changes the state of the cavity, which is described by a \emph{quantum Markov semigroup} in a certain coarse grained approximation described below (see \cite{Carbone2000} and ~\cite{Fagnola1994} for a mathematical overview, and~\cite{Englert2002} for the physical derivation of the master equation). 
In this section we give an intuitive description of the dynamics starting with a simplified discrete time model, with an emphasis on the statistics of measurements performed on the atoms. 


The cavity is described by a one mode continuous variable system with Hilbert space $\mathfrak{h} = \ell^{2}(\mathbb{N})$  whose canonical basis  vectors $(|e_{n}\rangle)_{n \geq 0}$ represent pure states of fixed number of photons. Therefore, if  $|\psi\rangle \in \mathfrak{h}$ is a pure state, the \emph{photon number distribution} of the cavity is given by $\abs{\inpr{e_{n}, \psi}}^{2}$. Mixed states are described by density operators, i.e. trace-class operators $\rho \in L^1(\mathfrak{h})$ which are  positive and normalised to have unit trace, and the observables are represented by self-adjoint elements of the von Neumann algebra of bounded operators $\mathcal{B}(\mathfrak{h})$ whose predual is $L^1(\mathfrak{h})$. 
Recall that the \emph{annihilation operator} $a$ on  $\mathfrak{h}$ is defined by
	\begin{equation*}
		a |e_{n}\rangle = \begin{cases} \sqrt{n} |e_{n-1}\rangle  &\text{if } n > 0  \\0 & \text{if } n = 0 \end{cases};
	\end{equation*}
its adjoint is the \emph{creation operator} $a^{*}$,
and $N = a^{*}a$ is the photon number operator such that $N|e_n\rangle  = n |e_n\rangle $. 
For every $\beta>0$, we introduce the following notation
\[D(N^\beta)=\left \{u=\sum_{n=0}^{+\infty} u_n |e_n\rangle: \sum_{n=0}^{+\infty} n^{2\beta} |u_n|^2<+\infty\right \}\]
for the domain of $N^\beta$ and we recall that $D(a)=D(a^*)=D(N^{\frac{1}{2}})$.
The atom is modelled by a two-dimensional Hilbert space 
$\mathbb{C}^{2}$ with standard orthonormal basis $\set{\ket{0},\ket{1}}$ consisting of the "ground" and "excited" states.  We denote by $\sigma^{*}$ and $\sigma$ the corresponding raising and lowering operators (i.e. $\sigma^{*} \ket{0} = \ket{1}$ etc.). The interaction between an atom and the cavity is described by the Jaynes-Cummings hamiltonian on 
$\mathbb{C}^2\otimes \mathfrak{h}$
	\begin{equation*}
		H_{\mbox{int}} = - g ( \sigma \otimes a^{*} + \sigma^{*} \otimes a),
	\end{equation*}
where $g$ is the coupling constant. 
The free hamiltonian is 
\begin{equation*}
		H_{\mbox{free}} = \omega \mathbf{1}\otimes   a^{*} a + \omega \sigma^{*} \sigma	\otimes \mathbf{1} ,
\end{equation*}
where $\omega$ is the frequency of the resonant mode; however by passing to the interaction picture the effect of the free evolution can be ignored. Therefore if the interaction lasts for a time $t_{0}$,  the joint evolution is described by the unitary operator $U:= \exp(-it_{0}H_{\mbox{int}})$ whose action on a product initial state is
	\begin{equation*}
		U :  \ket{1} \otimes \ket{k}  \mapsto \cos(\phi\sqrt{k+1}) \ket{1} \otimes \ket{k} + i \sin(\phi\sqrt{k+1})\ket{0}\otimes \ket{k+1} ,
	\end{equation*}
where $\phi:= t_{0} g$ is the \emph{accumulated Rabi angle}. If a measurement is performed on the outgoing atom in the standard basis, then the cavity remains in state $\ket{k}$ with probability $\cos^{2}(\phi\sqrt{k+1})$ or gains an excitation with probability $\sin^{2}(\phi\sqrt{k+1})$. By averaging over the outcomes, we obtain the cavity transfer operator 
$\mathcal{T}_{*}: L^1(\mathfrak{h})\to L^1(\mathfrak{h})$ 
	\begin{equation}\label{eq:transitionoperator}
		\mathcal{T}_{*}(\rho) = K_{1} \rho K_{1}^{\ast} + K_{2} \rho K_{2}^{\ast} = \mathcal{K}_{1*}(\rho) + \mathcal{K}_{2*}(\rho)
	\end{equation} 
where the \emph{Kraus operators} $K_{i}$ are given by
	\begin{equation*}
		\quad K_{1} = a^{*}\frac{\sin(\phi\sqrt{a a^{*}})}{\sqrt{a a^{*}}}, \quad K_{2} = cos(\phi\sqrt{a a^{*}}),
	\end{equation*}
and $\mathcal{K}_{i*}$ are the corresponding jump operators on the level of density matrices. Since each atom interacts with the cavity only once, the state of the cavity after $n$ such interactions is given by 
$\rho(n)= \mathcal{T}_{*}^{n}(\rho)$, which can be interpreted as a \emph{discrete time} quantum Markov dynamics. Let us imagine that after the interaction, each atom is measured in the standard basis and found to be either in the excited or the ground state. The master dynamics can be unravelled according to these events as follows
\begin{equation}\label{eq:discrete.unravel}
 \mathcal{T}_{*}^{n}(\rho) = \sum_{{\bf i} = (i_1,\dots ,i_n)} \mathcal{K}_{i_n*} \cdots \mathcal{K}_{i_1*} (\rho) 
\end{equation}
where each term of the sum represents the (unnormalised) state of the cavity after a certain sequence 
${\bf i}=(i_1,\dots ,i_n)\in\{0,1\}^n$ of measurement outcomes, whose probability is
$$
\mathbb{P}_\rho( i_1, \dots , i_n) = {\rm Tr} (\mathcal{K}_{i_n*} \cdots \mathcal{K}_{i_1*} (\rho) ).
$$ 
If $\Lambda_n({\bf i}):= \#\{ j: i_j= 0\}$ denotes the number of ground state atoms detected up to time $n$, we can use the previous relation to compute its moment generating function
\begin{equation}\label{eq:mgfdiscrete}
		\mathbb{E}_\rho\left(e^{s \Lambda_{n}}\right) = \sum_{k \geq 0} \mathbb{P}_\rho\left(\Lambda_{n} = k\right) e^{sk} = 
		\sum_{\bf i} e^{s\Lambda_n({\bf i})} {\rm Tr} (  \mathcal{K}_{i_n*} \cdots \mathcal{K}_{i_1*} (\rho) ) = {\rm Tr} ( \mathcal{T}_{*s}^n (\rho) )
\end{equation}
where 
	\begin{equation*}
		\mathcal{T}_{*s}(\rho) =e^{s} \mathcal{K}_{1*} (\rho)  +\mathcal{K}_{2*} (\rho)
	\end{equation*}
is a "deformed" transfer operator, i.e. a completely positive but not trace preserving map on $L^1(\mathfrak{h})$. The relation \eqref{eq:mgfdiscrete} and its continuous time analogue \eqref{eq:mgf} will be the key to analyse the large deviations properties of the counting process in terms of spectral properties of operators such as $\mathcal{T}_s$  and $\mathcal{L}_s$ below.	

To make the model more realistic we will pass to a continuous time description in which the incoming atoms are Poisson distributed in time with intensity $N_{\text{ex}}$, and the cavity is in contact with a thermal bath. If one ignores the details of short term cavity evolution, the discrete time dynamics can be replaced by coarse grained \emph{continuous time} Lindblad (master) equation \cite{Englert2002a}
	\begin{align}
		\frac{\id}{\id t} \rho(t) &= \mathcal{L}_{*}(\rho(t)),\nonumber\\
		\mathcal{L}_{*}(\rho) &= 
		\sum_{i=1}^{4} \left( L_{i} \rho L_{i}^{\ast} - \frac{1}{2}\lbrace L_{i}^{\ast} L_{i}, \rho \rbrace \right)\nonumber\\
		&=\sum_{i=1}^{4} L_{i} \rho L_{i}^{\ast}  + \mathcal{L}^{(0)}_{*} (\rho) = \sum_{i=1}^{4} \mathcal{J}_{i*} (\rho) + \mathcal{L}^{(0)}_{*} (\rho)
		\label{eq:lindblad}
	\end{align}
with jump operators $L_{i}$ defined by
	\begin{align}
		L_{1} &= \sqrt{N_{\text{ex}}} a^{*} \frac{\sin(\phi \sqrt{a a^{*}})}{\sqrt{a a^{*}}},\label{eq:L2}\\
		L_{2} &= \sqrt{N_{\text{ex}}} \cos(\phi \sqrt{a a^{*}}),\\
		L_{3} &= \sqrt{\nu+1}a,\\
		L_{4} &= \sqrt{\nu} a^{*}.\label{eq:L4}
	\end{align}
As before, the operators $L_{1}$ and $L_{2}$ are associated to the detection of an atom in the ground and excited state, respectively. The emission and absorption of photons due to contact with the bath is represented by operators $L_{3}$ and $L_{4}$, respectively. Between jumps the evolution is described by the semigroup $e^{t\mathcal{L}_*^{(0)}}(\rho):= e^{tG}(\rho)e^{tG}$ where
\be
\label{eq.def.G}
G:=-\frac{1}{2}\sum_{i=1}^4 L_i^*L_i=
-\frac{1}{2} \left(N_{\text{ex}}+\nu +(2\nu+1)N\right), \quad D(G)=D(N).
\ee

Since we deal with an infinite dimensional space and unbounded jump operators, the above definitions need to be formalised mathematically in order to insure existence and uniqueness of the different semigroups (see Proposition \ref{prop:def}). As it is customary in the theory of quantum dynamical semigroups with unbounded generator (\cite{Fagnola1999}), so far the generator $\mathcal{L}_*$ can be safely defined on the linear manifold generated by the operators $\ket{u}\bra{v}$ for $u,v \in D(G)$ (this manifold is also a core due to Proposition \ref{prop:def} and  \cite[Proposition 3.32]{Fagnola1999}) or, equivalently, we can interpret $\mathcal{L}$ applied to any $X \in \mathcal{B}(\mathfrak{h})$ as the sesquilinear form on $D(G)\times D(G)$ given by
\[\langle u, {\cal L}(X) v \rangle=\langle Gu,Xv\rangle + \langle u,XGv\rangle + \sum_{i=0}^4 \langle L_i u,XL_i v\rangle \quad \forall u,v \in D(G).
\]

\begin{defn}[\cite{Fagnola1999} and Section 3.1.2 in \cite{Bratteli1979}]
Let $\mathcal{B(\mathfrak{h})}$ be the space of bounded operators on $\mathfrak{h}$ 
endowed with the $w^*$- topology.
A quantum dynamical semigroup on $\mathcal{B(\mathfrak{h})}$ is a family $\mathcal{S} = (\mathcal{S}(t) )_{t\geq 0}$ of bounded operators on $\mathcal{B(\mathfrak{h})}$ with the following properties
\begin{itemize}
    \item[(i)]
$\mathcal{S}(0) = I$, 
\item[(ii)]
$\mathcal{S}(s+t) = \mathcal{S}(s)\mathcal{S}(t)$ for all $s,t\geq 0$,  

\item[(iii)] $\mathcal{S}(t)$ is completely positive for all $t\geq 0$,

\item[(iv)] $\mathcal{S}(t)$ is $w^*$-continuous operator on $\mathcal{B(\mathfrak{h})}$ for all $t\geq 0$,

\item[(v)] for each $X\in\mathcal{B(\mathfrak{h})}$, the map $t\mapsto \mathcal{S}(t)(X)$ is continuous with respect to the $w^*$-topology on $\mathcal{B(\mathfrak{h})}$.
\end{itemize}

The dynamical semigroup 
$\mathcal{S}(t)$ is called Markov (sub-Markov) if ${\cal S}(t)(\mathbf{1})=\mathbf{1}$ (${\cal S}(t)(\mathbf{1})\leq \mathbf{1}$) holds true for every time $t$.

\end{defn}

We recall that since the maps $\mathcal{S}(t)$ are positive, the fact that they are $w^*$-continuous is equivalent for them to be normal (\cite[Lemma 2.4.19 and Theorem 2.4.21]{Bratteli1979}). The $w^*$-generator $\mathcal{Z}$ is the operator defines as
	\begin{equation*}
		 \mathcal{Z}(X) := w^*-\lim_{h \downarrow 0} \tfrac{1}{h}\left( \mathcal{S}(h) (X) - X\right)
	\end{equation*}
for all $X \in D(\mathcal{Z}):= \left \{X \in \mathcal{B}(\mathfrak{h}): \exists \, w^*-\lim_{h \downarrow 0} \tfrac{1}{h}\left( \mathcal{S}(h) (X) - X\right) \right \}$, which is a $w^*$-dense linear space of $\mathcal{B}(\mathfrak{h})$.
Although no simple expression exists for the operators $\mathcal{S}(t)$ in terms of the generator $\mathcal{Z}$, it is helpful to think of $\mathcal{S}(t)$ as the exponential of the generator 
	\begin{equation}\label{eq:semigroup}
		\mathcal{S}(t) (X) = e^{t \mathcal{Z}} (X),
	\end{equation}
especially from the point of view of relating spectral properties of $\mathcal{Z}$ to those of $\mathcal{S}(t)$, e.g. \emph{spectral mapping theorems}. In Proposition \ref{prop:def} below we show that the Heisenberg picture Lindbladian $\mathcal{L}$ is the generator of a quantum Markov semigroup $(\mathcal{T}(t))_{t \geq 0}$ on $\mathcal{B}(\mathfrak{h})$; we postpone the proof of Proposition \ref{prop:def} to the Appendix \ref{app:profs.prop1&2}.
\begin{prop} \label{prop:def}
\begin{enumerate}
\item $\mathcal{L}$ generates a unique quantum Markov semigroup $\mathcal{T}(t)$ which has the following integral representation: for every $X \in \mathcal{B}(\mathfrak{h})$,
\be \label{eq:integral1}
\begin{split}
&\mathcal{T}(t)(X) = e^{t\mathcal{L}^{(0)}} (X)+\\
&+ \sum_{k\ge 1}\sum_{i_1,..i_k=1}^4  \int_{0\le t_1\le\cdots \le t_k\le t} X(t; t_1,i_1,\dots, t_k, i_k)
dt_1\dots dt_k\\
\end{split}
\ee
and
\[
X(t; t_1,i_1,\dots, t_k, i_k):= e^{(t-t_k)\mathcal{L}^{(0)}}{\mathcal J}_{i_k}\cdots  e^{(t_2-t_1)\mathcal{L}^{(0)}}{\mathcal J}_{i_1}e^{t_1\mathcal{L}^{(0)}}(X),
\]
where the equality is understood in terms of the associated bilinear form $\langle u, \mathcal{T}(t)(X)v\rangle$ for  $ u,v \in \mathfrak{h}$.
\item $\left(\mathcal{T}(t)\right)_{t \geq 0}$ has a unique faithful stationary state
	\begin{equation} \label{eq:stationary}
		\rho_{\text{ss}} := 
			  \rho_{\text{ss}}(0) \sum_{n \ge 0} \prod_{k=1}^{n}\left( \frac{\nu}{\nu +1} + \frac{N_{\text{ex}}}{\nu +1} \frac{\sin^{2}(\phi \sqrt{k})}{k} \right) |e_n \rangle \langle e_n|
	\end{equation}
with $\rho_{\text{ss}}(0)$ taken such that $\Tr (\rho_{\text{ss}}) =1$.
\item $\left(\mathcal{T}(t)\right)_{t \geq 0}$ is ergodic, in the sense that any initial state $\rho$ converges to the stationary state
$$
w-\lim_{t\to\infty} \mathcal{T}_{*}(t)(\rho)  = \rho_{ss}.
$$ 
\end{enumerate}
\end{prop}

\begin{figure}[t] 
    \centering
    \includegraphics[width=0.9\textwidth]{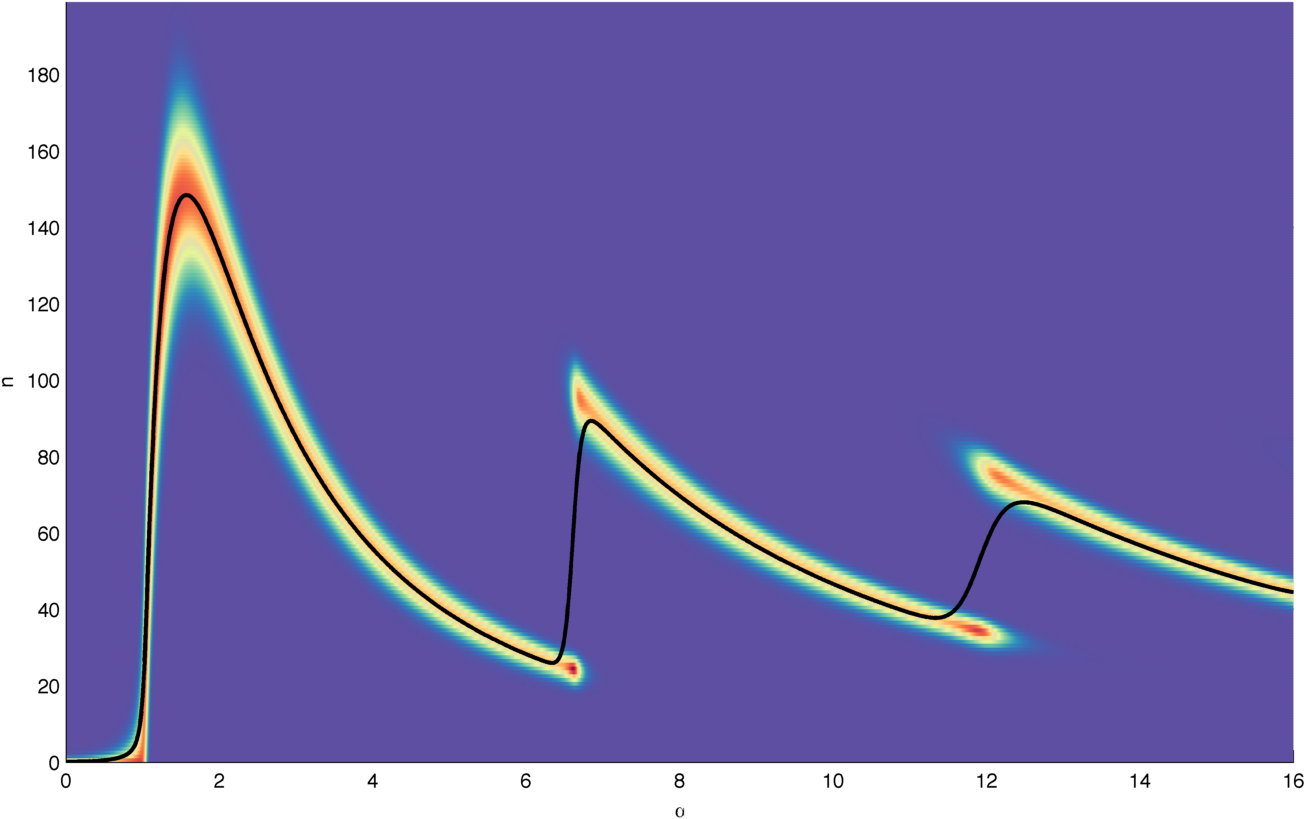}
    \caption{Mean photon number (black line) and photon number distribution (background) in the stationary state $\rho_{ss}$ as function of $\alpha = \sqrt{N_{\mbox{ex}}} \phi$}\label{fig:stationary}
\end{figure}

The dependence of the stationary mean photon number and photon number distribution on the "pumping parameter" 
$\alpha:= \sqrt{N_{\text{ex}  }}\phi$ is shown in Fig. \ref{fig:stationary}, for $\nu = 0.15$ and $N_{\text{ex}} = 150$. We note two interesting features in this figure: first, there is a sharp change in the mean photon number at $\alpha \approx 1$ followed by less pronounced jumps near $\alpha = 6.66 $ and $\alpha = 12$.
The other, related, feature to note is that the photon number distribution has a single peak for most values of $\alpha$  except in certain regions such as around the critical point $\alpha \approx 6.66$, where the stationary state has two local maxima. 
We will come back to these aspects in the next section and show that they are related to features of the counting trajectories such as intermittency which indicates proximity to a dynamical phase transition. The reason for plotting the stationary distribution in terms of $\alpha$ (with $N_{\text{ex}}$ fixed) rather that $\phi$ is because the transitions appear to occur at fixed values of $\alpha$ and sharpen as $N_{\text{ex}}\to \infty$. This will be investigated further in section \ref{sec.numerics}.

%

\subsection{The counting process and the deformed transition operator}

To better understand the behaviour of the stationary state illustrated in Figure \ref{fig:stationary}, we unravel the Markov semigroup $\mathcal{T}_*(t)$ with respect to the four counting processes associated to the jump terms (\ref{eq:L2} - \ref{eq:L4}), each of them corresponding to a counting measurement of the quantum output process. If $\rho$ is the initial state of the cavity, then $\rho(t):=\mathcal{T}_*(t)(\rho)$ is 
 the evolved state at time $t$ which (in analogy to Eq.\eqref{eq:discrete.unravel}) can be seen as an average over all possible counting events in the environment 
\begin{equation}\label{eq:unravel}
\begin{split}
\rho(t):=&\mathcal{T}_*(t)(\rho) = e^{t \mathcal{L}^{(0)}_{*}} (\rho)+\\
& \sum_{k\geq 1} \, \sum_{i_1,\dots i_k =1}^4 \int_{0\le t_1\le\cdots \le t_k\le t}
\rho(t; t_1,i_1\dots ,t_k, i_k) dt_1 \dots dt_k\\
\end{split}
\end{equation}
where the integrand
\begin{equation*}\label{eq.cond.state.unnorm}
\rho(t; t_1,i_1,\dots, t_k, i_k):= 
e^{(t-t_k) \mathcal{L}^{(0)}_{*}} \mathcal{J}_{i_{k*}} \cdots e^{(t_2- t_1) \mathcal{L}^{(0)}_{*}}  \mathcal{J}_{i_1*} e^{t_1 \mathcal{L}^{(0)}_{*}} (\rho),
\end{equation*}
is the unnormalised state of the cavity given that detections of type $i_1,\dots , i_k\in\{1,2,3,4\}$ have occurred at times 
$0\leq t_1\leq \dots \leq t_k\leq t$, and no other counting events happened in the meantime. Its trace 
is interpreted as the probability of observing the given measurement record.  Notice that equation (\ref{eq:unravel}) can be obtained by duality from equation (\ref{eq:integral1}), so this description is mathematically rigorous. Among the four counting 
processes we focus on the first one associated with the detection of an atom in the ground state and simultaneous absorption of a photon by the cavity. We denote by $\Lambda_t$ the \emph{total number} of such atoms detected up to time $t$: for every $n_1 \in \mathbb{N}$
\[{\mathbb P}_\rho(\Lambda_t=n_1)=\sum_{n_2,n_3,n_4 \ge 0}
\sum_{(***)}  
\int_{0\le t_1\le\cdots \le t_k\le t} 
{\rm tr}(\rho(t; t_1,i_1,\dots, t_k, i_k))\;
dt_1\dots dt_k
\]
where $(***)$ stands for $$\{i_1,..i_{n_1+\cdots +n_4}=1,\dots,4: \, \#\{k\,:\, i_k=j\}=n_j \,\forall\, j=1,\dots,4\}.
$$
Similarly to the discrete case, by using the above unravelling and point 3 in Proposition \ref{prop:defdef} below, we can show  that the moment generating function of $\Lambda_t$ is given by 
\begin{equation}\label{eq:mgf}
		\mathbb{E}_\rho\left(e^{s \Lambda_{t}}\right) =
		\Tr \left(  \mathcal{T}_{*s}(t)(\rho)\right)=
		  \Tr \left( \rho \mathcal{T}_s(t)(\mathbf{1})\right).
	\end{equation}
where $\left(\mathcal{T}_s(t)\right)_{t \geq 0}$ is the quantum dynamical semigroup on $\mathcal{B}(\mathfrak{h})$ with generator 
	\begin{equation}
		\mathcal{L}_{s}(X) = 
		 e^s \mathcal{J}_1(X) + \sum_{i=2}^{4} \mathcal{J}_j(X)+ \mathcal{L}^{(0)}(X)= 
		(e^s-1) \mathcal{J}_1(X)+  \mathcal{L}(X) ,\label{eq:perturbed}
	\end{equation}
and $\left(\mathcal{T}_{*s}(t)\right)_{t \geq 0}$ is the predual semigroup on $L^1(\mathfrak{h})$. 
This is formalised in the following Proposition whose proof can be found in Appendix \ref{app:profs.prop1&2}.
\begin{prop} \label{prop:defdef}
For all $s \in \mathbb{R}$, $\mathcal{L}_s$ generates a semigroup $\mathcal{T}_s = {(\mathcal{T}_s(t))}_{t\ge 0}$ such that
\begin{enumerate}
\item $\mathcal{T}_s(t)$ is a quantum dynamical semigroup.\\
\item $\mathcal{T}_s(t)$ is the unique solution to
\begin{equation} \label{eq:defsol}
\langle u,\mathcal{T}_s(t)(X) v \rangle=\langle e^{tG}u, X e^{tG}v\rangle+ \sum_{i=1}^4 \int_0^t \langle L^\prime_i e^{rG}u, \mathcal{T}_s(t-r)(X)L^\prime_i e^{rG}v \rangle dr
\end{equation}
for every $u,v \in \mathfrak{h}$, where $L^\prime_1=e^{s/2}L_1$ and $L^\prime_i=L_i$ for $i=2,3,4$.
\item $\mathcal{T}_s(t)$ has the integral representation
\begin{equation} \label{eq:intrepr}
\begin{split}
&\mathcal{T}_s(t)(X)= e^{t\mathcal{L}^{(0)}} (X) +
\\
& +\sum_{k\ge 1}\sum_{i_1,..i_k=1}^4  \int_{0\le t_1\le\cdots \le t_k\le t} X_s(t; t_1,i_1,\dots, t_k, i_k) dt_1\dots dt_k,
\end{split}
\end{equation}
for every $t\geq 0$, $X \in B(\mathfrak{h})$ and
\[
X_s(t; t_1,i_1,\dots, t_k, i_k):= e^{(t-t_k)\mathcal{L}^{(0)}}{\mathcal J}^\prime_{i_k}\cdots  e^{(t_2-t_1)\mathcal{L}^{(0)}}{\mathcal J}^\prime_{i_1}e^{t_1\mathcal{L}^{(0)}}(X),
\]
\sloppy where the equality has to be read for the associated bilinear form $\langle u, \mathcal{T}_s(t)(X)v\rangle$ for  $ u,v \in \mathfrak{h}$ and ${\cal J}^\prime_1=e^s{\mathcal J}_1$, 
 ${\cal J}^\prime_i={\cal J}_i$ for $i=2,3,4$.
\end{enumerate}
\end{prop}

Equation \eqref{eq:mgf} plays a central role in this paper; we will use it to formulate a large deviations principle for the counting process $\Lambda_{t}$, and in particular, to relate the moment generating function of $\Lambda_{t}$ to the spectral properties of  $\mathcal{L}_{s}$. 
Note that $\mathcal{L}_{s}$ differs from the Lindblad generator by the factor $e^s$ multiplying the jump term associated to the detection of a ground state atom. It is still the generator of a completely positive semigroup, but it is no longer identity preserving, and therefore does not represent a physical evolution except for $s=0$. 


\begin{figure}[t]
    \centering
    \includegraphics[width=0.7\textwidth]{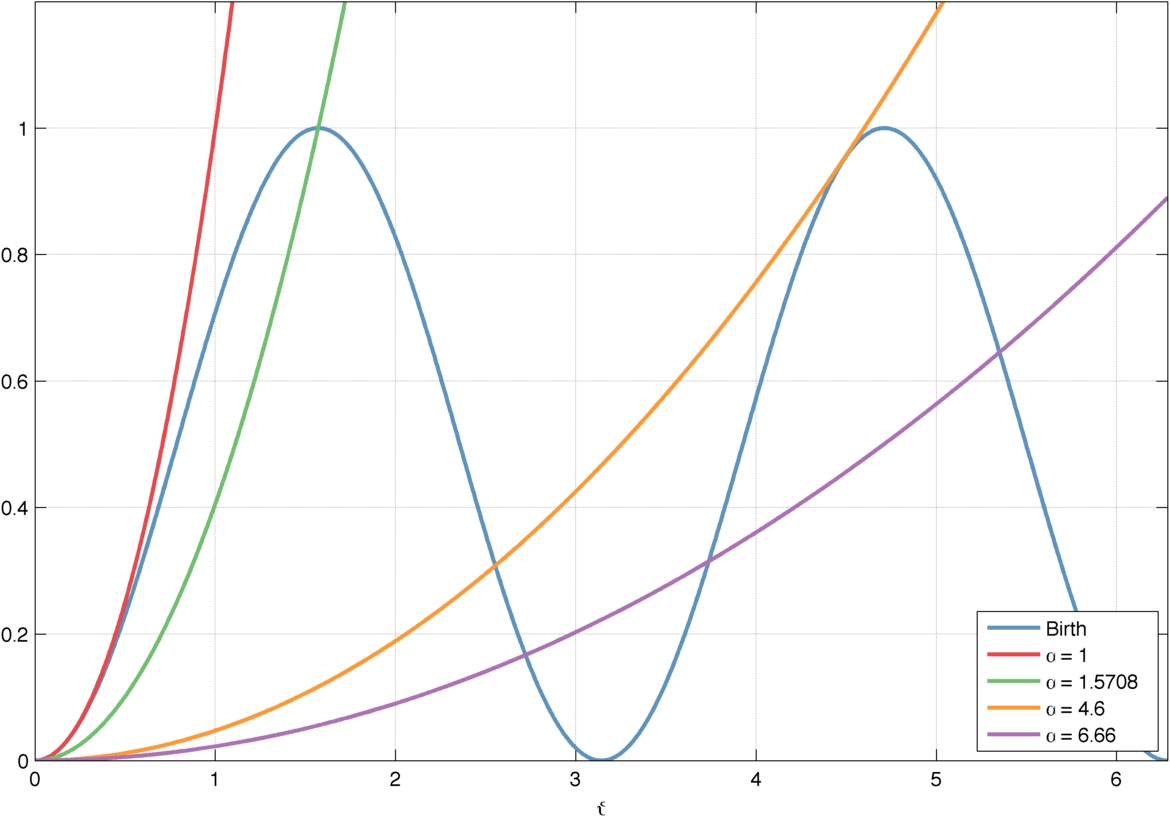}
    \caption{The birth (blue) and death rates as functions of $\vartheta$ for different values of $\alpha$. The intersection points correspond to minima and maxima of the stationary distribution.}\label{fig:b.d.rates}
\end{figure}

The unravelling \eqref{eq:unravel} allows for a classical 
interpretation of the cavity dynamics. Indeed, the semigroup generated by 
$\mathcal{L}$ (and $\mathcal{L}_s$) leaves invariant the commutative subalgebra $\mathcal{B}_d(\mathfrak{h})\subset \mathcal{B}(\mathfrak{h})$ generated by the number operator $N$, and the restriction of $\left(\mathcal{T}(t)\right)_{t \geq 0}$ to the diagonal algebra is the dynamical semigroup of a classical  \emph{birth-death process} on the state space $\{0, 1,2, \dots\}$, with rates 
\begin{eqnarray} \label{eqn:rates}
&&
\lambda_{k}^2:= N_{ex} \sin(\phi \sqrt{k+1})^{2} + \nu (k+1) ,\quad k\geq 0\nonumber \\
&&
\mu_{k}^2 := (\nu+1)k, \quad k\geq 1.\label{eq.birth-death}
\end{eqnarray}

Figure \ref{fig:b.d.rates} shows the birth and death rates (minus the common factor $\nu k$) rescaled by a factor $N_{ex}$, in the limit $N_{ex} \rightarrow \infty$, as functions of the parameter $\vartheta:= \sqrt{(k+1)/N_{ex}} \alpha= \phi\sqrt{k+1} $. In this regime the rates become the functions $\lambda^2_\vartheta = \sin(\theta)^2$ and $\mu^2_\vartheta= \alpha^{-2 }\vartheta^2$ of the \emph{continuous parameter} $\vartheta$, and we plot $\lambda^2_\vartheta$  along with $\mu^2_\vartheta$ for different values of $\alpha$. The intersection points correspond to minima and maxima of the stationary distribution \cite{Englert2002a} as suggested by the following argument. For $\alpha<1$ the death rate is always larger than the birth rate and the distribution is maximum at the vacuum state. For $1< \alpha<4.6$ there is a single non-trivial intersection point such that the birth rate is larger to its left and smaller to its right, and therefore corresponds to the maximum of the stationary distribution. Similarly, when $4.6<\alpha< 7.8$ the rates intersect in three points, the first and last are located at local maxima while the middle  point is a local minimum, so we deal with a bimodal distribution. 
However, while this analysis clarifies the emergence of multimodal distributions, it does not explain the sudden jump of the mean photon number at $\alpha\approx 6.66$, and higher values. 


This feature can be intuitively understood by appealing to the effective potential 
model \cite{Haroche&Raimond}. If we think of the photon number as a continuous variable and introduce  a fictitious potential $U$ defined by
\begin{equation}\label{eq.potential}
\rho_{ss}(n) = \rho_{ss}(0) e^{-U(n)}, 
\end{equation}
then the photon number distribution appears as the thermal equilibrium distribution of a particle moving in the potential $U$ (with $k_B \cdot T=1 $), see Figure \ref{fig:potentials}. When the potential has a single local minimum (for $0<\alpha<4.6$), the stationary distribution is unimodal and concentrates around this point. The cavity state fluctuates around the mean, and $\Lambda_t$ increases steadily with average rate. When there are two (or more) local minima of different height, the higher minimum  corresponds to a metastable phase from which the system eventually escapes due to thermal fluctuations. The rate of return to the metastable phase is typically much lower due to the larger potential barrier that needs to be climbed. The point $\alpha \approx 6.66$ where the two local minima are equal plays the role of a "phase transition", and corresponds roughly to the point where the mean photon number changes abruptly. Here the cavity spends long periods of time around the two local maxima with rare but quick transitions between them. The change from the low energy to the high energy mode is accompanied by a clear change in the slope of the counting process $\Lambda_t$.

\begin{figure}[t]
    \centering
        \begin{subfigure}[b]{0.95\textwidth}
          \centering
          \begin{subfigure}[b]{0.3\textwidth}
                  \centering
                  \includegraphics[width=\textwidth]{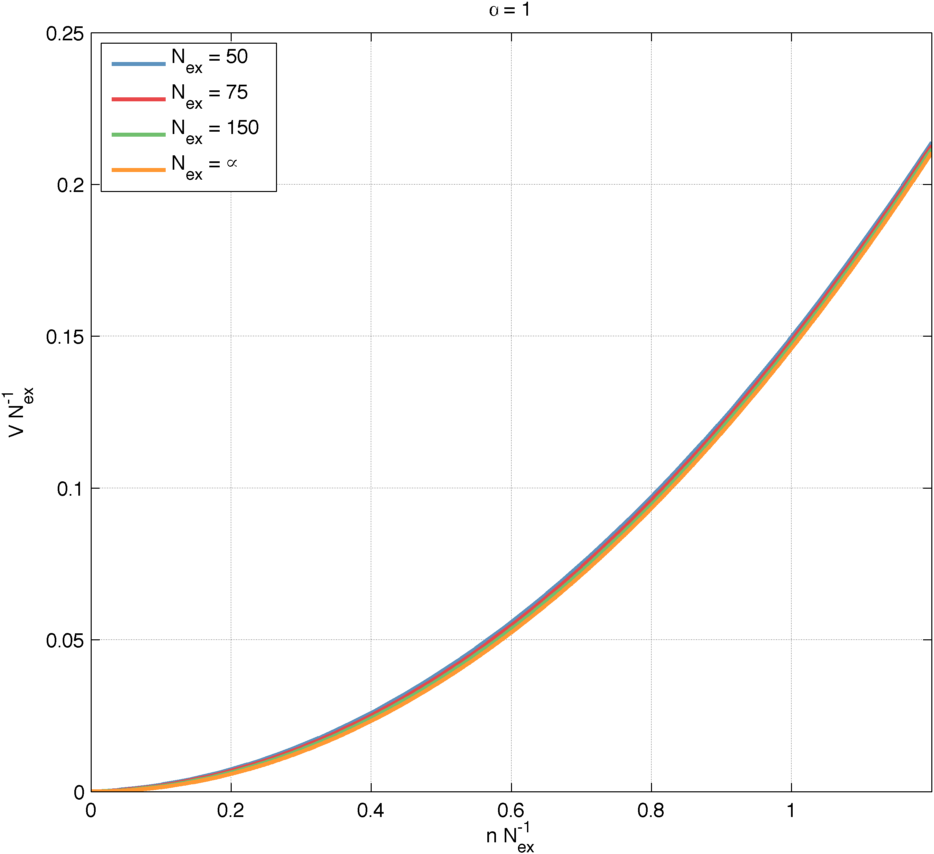}
          \end{subfigure}%
          \quad
                    \begin{subfigure}[b]{0.3\textwidth}
                  \centering
                  \includegraphics[width=\textwidth]{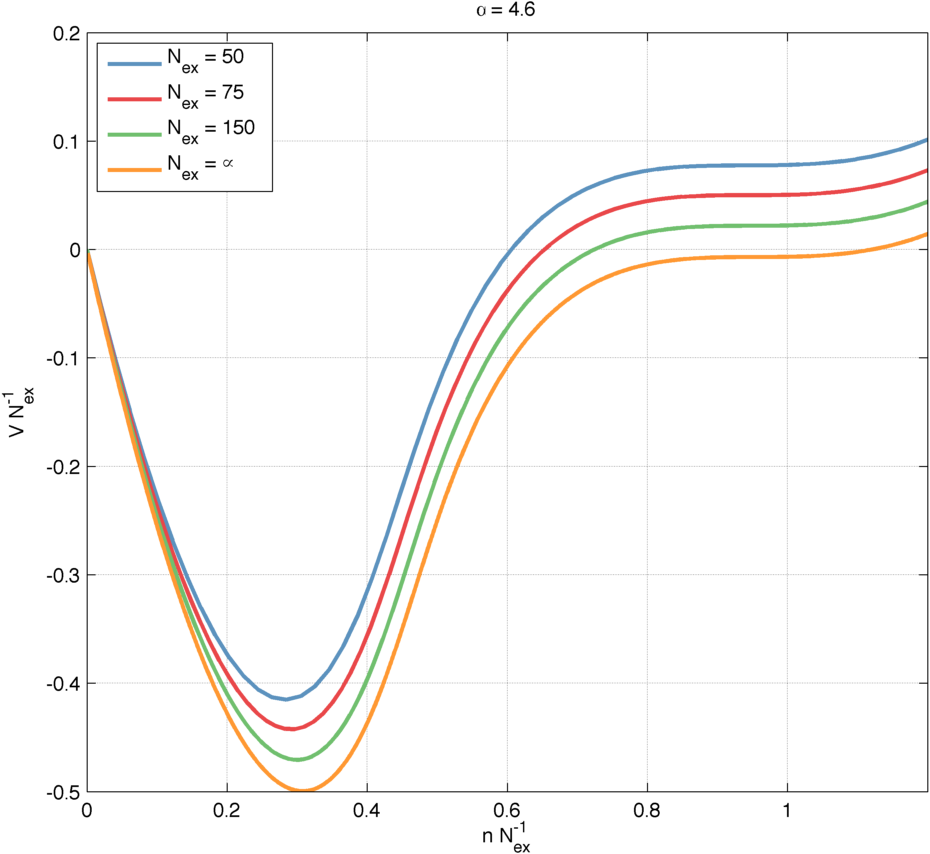}
          \end{subfigure}
          \quad
                              \begin{subfigure}[b]{0.3\textwidth}
                  \centering
                  \includegraphics[width=\textwidth]{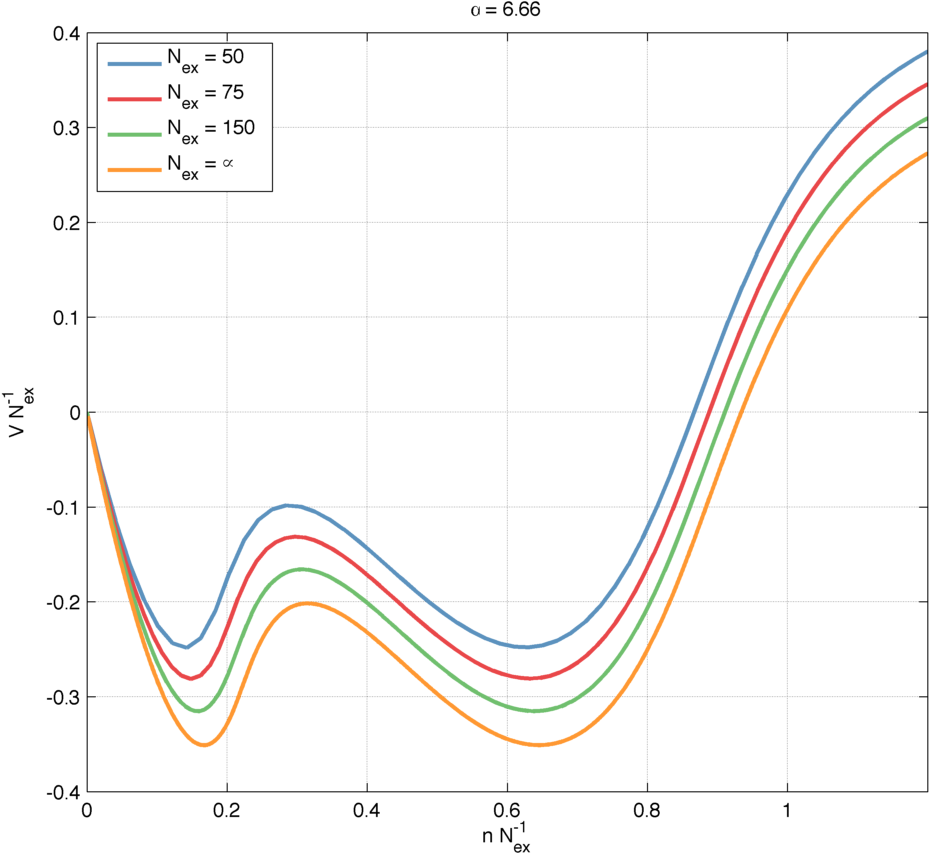}
          \end{subfigure}
    \end{subfigure}
    \caption{Rescaled potentials $U(n)/N_{\text{ex}}$ as function of $n/N_{\text{ex}}$, for various finite $N_{\text{ex}}$ converge to a limit potential for $N_{\text{ex}} \rightarrow \infty$. For $\alpha < 1$ the potential is minimum at zero; for  $1<\alpha <4.6$ it has a unique minimum away from $n=0$; for $4.6<\alpha< 7.8$ there are two local minima which become equal at $\alpha \approx 6.66$.}\label{fig:potentials}
\end{figure}

In the stationary regime the mean 
$\mathbb{E}_{\rss}( \Lambda_t)$ grows linearly with time with rate ${\rm Tr} (\rho_{ss} L_1^*L_1)$. This expression can be obtained by differentiating the moment generating function \eqref{eq:mgf} at $s=0$
\begin{eqnarray*}
\frac{\mathbb{E}_{\rho}( \Lambda_t)}{t} &=&
\frac{1}{t}\left.\frac{d}{ds}\mathbb{E}_{\rho}(e^{s\Lambda_t})\right|_{s=0} = 
\frac{1}{t}\left.\frac{d}{ds}
{\rm Tr}(\rho \mathcal{T}_s (t)(\mathbf{1}))
\right|_{s=0}\\
&=& \frac{1}{t}\left. \int_{u=0}^t du
{\rm Tr}(\rho 
\mathcal{T}_s (u) \circ \mathcal{J}_1 \circ 
\mathcal{T}_s (t-u) (\mathbf{1}))\right|_{s=0}\\
&=&
\frac{1}{t} \int_{u=0}^t du
{\rm Tr}(
\mathcal{T}_*(u)(\rho)  L_1^* L_1)
\end{eqnarray*}
where we used the fact that $\left. \frac{d \mathcal{L}_s}{ds}  \right |_{s=0} = \mathcal{J}_1$, cf. \eqref{eq:perturbed}. The rate is then obtained by taking $t\to \infty$ and using the fact that $\rho$ converges to the stationary state.

Using the property of the birth-death process
$\sum_n \rss(n) (\lambda^2_n- \mu^2_n)=0$ we can further write the rate as 
\be \label{eq.statmean}
\frac{\mathbb{E}_{\rss} ( \Lambda_t) }{t} =N_{ex} \sum_n  \rho_{\text{ss}}(n) \sin^2(\phi \sqrt{n+1})= 
\sum_n n  \rho_{\text{ss}}(n) - \nu.
\ee
Unlike the "first order transition" occurring at $\alpha=6.66$, a "second order transition" occurs at $\alpha\approx 1$. Here the first derivative of the mean photon number has a jump in the limit of $N_{\text{ex}}\to\infty$. This and the scaling of the potential $U$ with 
$N_{\text{ex}}$ will be discussed in section \ref{sec.numerics}.

Equation \eqref{eq.statmean} shows that the statistics of the trajectories are therefore closely related to the dynamics of the cavity and consequently with its stationary state. The next step is to think of the time trajectories as "configurations" of  stochastic system draw from ideas in non-equilibrium statistical mechanics and large deviations theory to study their 
phases and phase transitions. 
%
%

\subsection{Large deviations}

	The main result of this paper is the existence of a \emph{large deviations principle} for the counting process $\Lambda_{t}$ introduced above. Such results have already been obtained in the context of discrete time quantum Markov chains with finite dimensional systems \cite{Hiai2007}, but the novelty here is that we consider a continuous time Markov process with an infinite-dimensional system. The physical motivation lies in the new approach to the study of phase transitions for open systems developed in \cite{Garrahan2009, Garrahan2011}. Here the idea is to identify \emph{dynamical phase transitions} of the open system, by analysing the statistics of jump trajectories in the long time (stationary) regime. The trajectories play an analogous role to the configurations of a statistical mechanics model at equilibrium. In this analogy, the parameter $s$ of the moment generating function \eqref{eq:mgf} can be seen as a "field" which biases the distribution of trajectories in the direction of active or passive trajectories by effectively changing the probability of a trajectory $\omega:= (i_1, t_1, \dots,  i_n, t_n)$ by a factor $\exp(s \Lambda_t(\omega))$. When $\alpha$ is such that the stationary distribution is unimodal, the trajectories' distribution changes smoothly from passive ones for $s<0 $ to active ones for  $s>0 $. However, near $\alpha\approx 6.66$ (corresponding to the jump in the mean photon number) there is a steep change in the counting rates around $s=0$. The active trajectories are associated to periods when the cavity is in the higher, excited phase while the passive trajectories are connected to the lower phase. Since the cavity makes very rare transitions between the phases, any trajectory -- when followed for long but finite periods of time -- falls typically into one of the two distinct categories (see Figure \ref{fig:trajectories}). Our goal is to investigate whether this distinction survives the infinite time limit, in which case we would deal with a dynamical phase transition characterised by the non-analyticity of a certain large deviations rate function. We will show that this is not the case, but rather we deal with a \emph{cross-over} behaviour; that is, the count rate does not jump but has a very steep change around $s=0$, which appears to become a jump only in the limit of infinite pumping rate $N_{ex} \to\infty$ (see Section 5).

\begin{figure}[t]
    \centering
        \begin{subfigure}[b]{0.95\textwidth}
          \centering
          \begin{subfigure}[b]{0.31\textwidth}
                  \centering
                  \includegraphics[width=\textwidth]{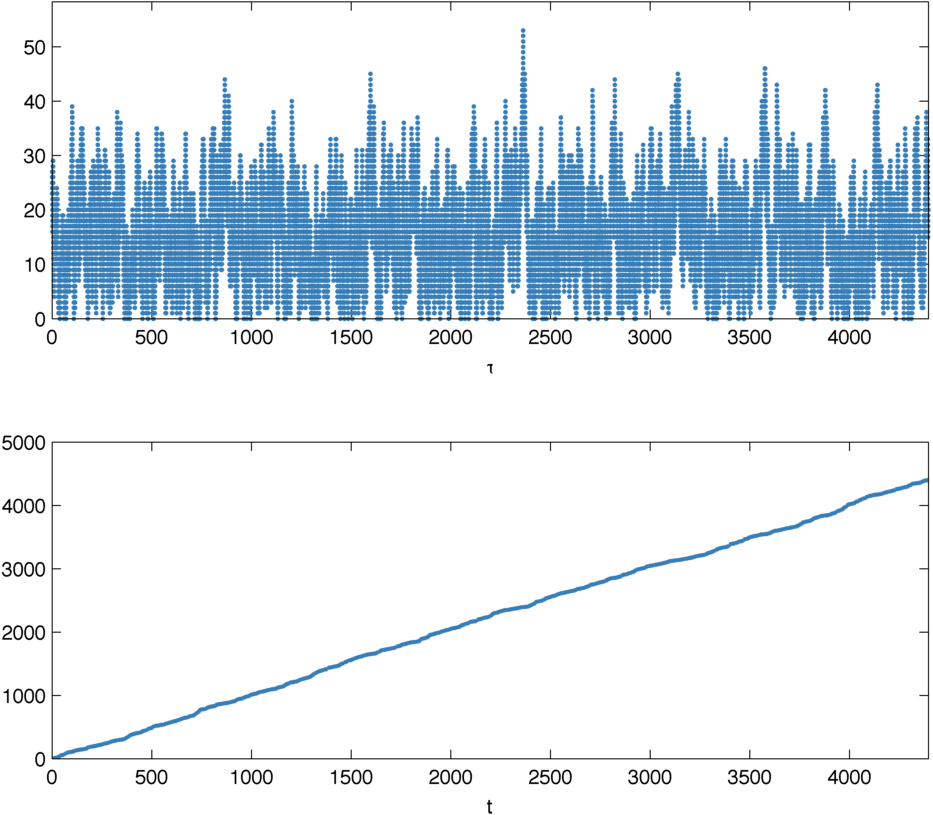}
          \end{subfigure}%
          \quad
          \begin{subfigure}[b]{0.31\textwidth}
                  \centering
                  \includegraphics[width=\textwidth]{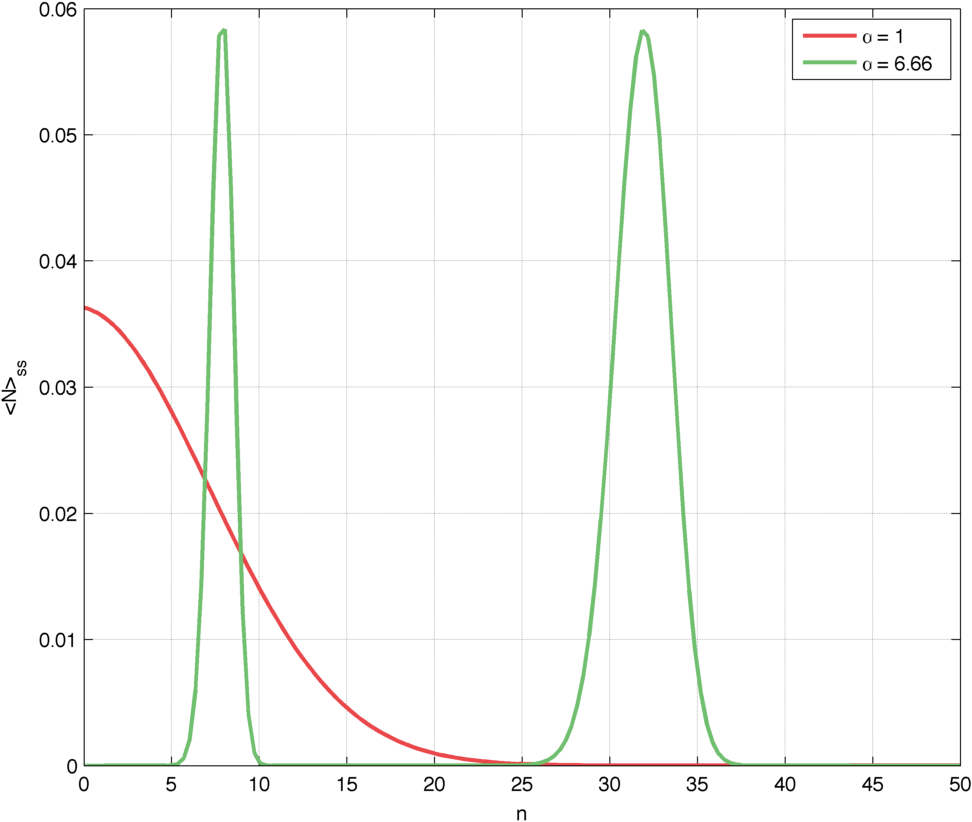}
          \end{subfigure}
          \quad
                    \begin{subfigure}[b]{0.31\textwidth}
                  \centering
                  \includegraphics[width=\textwidth]{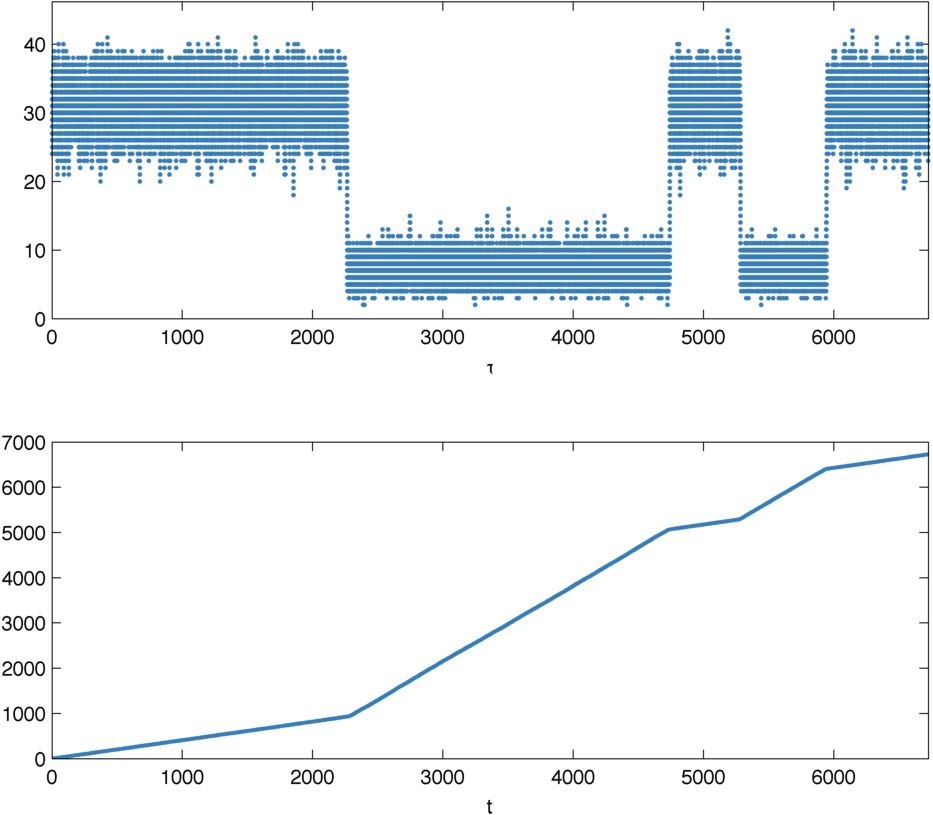}
          \end{subfigure}
    \end{subfigure}
    \caption{Sample trajectories for the birth-death process describing the cavity state jumping on the ladder of Fock states $|k\rangle\langle k|$ (top left and right) and total measurement counts $\Lambda_{t}$ (bottom left and right) for 
    $\alpha \approx 1$ (left) and $\alpha \approx 6.66$ (right) at $N_{\text{ex}}=50$. The corresponding to stationary state distributions (center) showing large variance at $\alpha \approx 1$ (red) and bistability at $\alpha \approx 6.66$ (green).}\label{fig:trajectories}
\end{figure}

We will now briefly review some basic notions of large deviations needed in the paper. We refer the reader to \cite{Dembo2010} for a complete treatment; see \cite{Ellis1995} for a comprehensive overview and \cite{Touchette2009} for an introduction to large deviations in the context of statistical mechanics. Large deviations is a framework for studying rare events, more precisely events whose probabilities decay exponentially for a sequence of probability distributions. A key result is the \emph{G\"{a}rtner-Ellis theorem}, which relates the rate of the exponential decay to the limiting behaviour of the moment generating functions associated to the random variables. 

Informally, a sequence $(\mu_{n})_{n \in \mathbb{N}}$ of probability distributions on $\mathbb{R}^d$ endowed with the Borel $\sigma$-field $\mathcal{B}$ satisfies a \emph{large deviation principle} (LDP)  if there exists a function $I : \mathbb{R}^d \rightarrow [0,\infty]$ such that 
	\begin{equation*}
		\mu_{n}(\id x) \approx e^{-n I(x)} \id x.
	\end{equation*}
More rigorously, the function $I$ is called a \emph{rate function} if it is lower semicontinuous (that is, its level sets $\set{x \in \mathbb{R}^d : I(x) \leq \alpha}$ are closed); if in addition its level sets are compact, we call it a \emph{good rate function}. The domain of $I$ is the set of points in $\mathbb{R}^d$ for which $I$ is finite. The limiting behaviour of the probability measures $\set{\mu_{n}}$ is characterised in terms of asymptotic upper and lower bounds on the values that $\mu_{n}$ assigns to measurable subsets $\Gamma \in \mathcal{B}$. The sequence of probability measures $\set{\mu_{n}}$ satisfies a large deviation principle with a rate function $I$ (or shortly, satisfies an LDP) if for all $\Gamma \in \mathcal{B}$,
	\begin{equation}\label{eq:ldp.def}
		-\inf_{x \in \Gamma^0} I(x) \leq \liminf_{n \rightarrow \infty} \tfrac{1}{n} \log \mu_{n}(\Gamma) \leq \limsup_{n \rightarrow \infty} \tfrac{1}{n} \log \mu_{n}(\Gamma) \leq -\inf_{x \in \bar{\Gamma}} I(x).
	\end{equation}

Our goal is to prove an LDP for the counting process $\Lambda_{t}$ of the atom maser; we will do this not by showing that $\Lambda_{t}$ satisfies the above definition directly, but by applying the G\"{a}rtner-Ellis theorem, which gives sufficient conditions on the sequence of probability measures in order to satisfy an LDP.

\begin{thm}[G\"{a}rtner-Ellis theorem\cite{Dembo2010}, \textbf{pp}. 44-55]
Let $(Z_{n})_{n\in \mathbb{N}}$ be a sequence of random variables in $\mathbb{R}^{d}$ with laws $\mu_{n}$. Suppose that the (limiting) logarithmic moment generating function
\begin{equation*}
	\lambda({\bf s}) = \lim_{n \rightarrow \infty} \tfrac{1}{n} \log \mathbb{E}\left[ e^{\langle n {\bf s}, Z_{n} \rangle} \right],\quad {\bf s} \in \mathbb{R}^{d}
\end{equation*}
exists as an extended real number and is finite in a neighbourhood of the origin, and let $\lambda^{\ast}$ denote the Fenchel-Legendre transform of $\lambda$, given by
\begin{equation*}
	\lambda^{\ast}(x) = \sup_{\lambda \in \mathbb{R}^{d}} \set{\langle {\bf s}, x\rangle - \lambda({\bf s})}.
\end{equation*}
If $\lambda$ is a essentially smooth, lower semicontinuous function (e.g. $\lambda$ is differentiable on $\mathbb{R}^{d}$) then $(Z_{n})_{n\in \mathbb{N}}$ satisfies a LDP with good rate function $\lambda^{\ast}$.
\end{thm}

The discrete index in the G\"{a}rtner-Ellis Theorem can be replaced by a continuous one with the obvious 
modifications in \eqref{eq:ldp.def}. By the G\"{a}rtner-Ellis Theorem, $\Lambda_t$ satisfies an LDP if the following limit exists and is a differentiable function, 
\begin{equation}
		\label{eq:mfgconvergence} \lambda(s) := \lim_{t \rightarrow \infty} \frac{1}{t} \log \mathbb{E}_\rho\left(e^{s \Lambda_{t}}\right)=\lim_{t \rightarrow \infty} \frac{1}{t} \log  \Tr \left( \rho \mathcal{T}_s(t)(\mathbf{1})\right).
	\end{equation}
We will show that this is indeed true and $\lambda(s)$ is \emph{spectral bound} (i.e. the eigenvalue with the largest real part) of a certain generator $L^{(d)}_{s}$ which is closely related to $\mathcal{L}_s$. An essential ingredient is the Krein-Rutman theory, that generalize Perron-Frobenius Theorem to compact positive semigroups and ensures that  $\lambda(s)$ is real and non-degenerate. In particular our analysis shows that $\lambda(s)$ is smooth and its derivatives at $s=0$ are the limiting cumulants of $\Lambda_t$
$$
\lim_{t\to\infty} \frac{1}{t} C_k(\Lambda_t) = \left. \frac{d^k \lambda(s)}{ds^k} \right|_{s=0}, \quad k\geq 1,
$$
the first two being the mean and the variance. Moreover the generator $L^{(d)}_s$ has a non-zero spectral gap; this spectral analysis is illustrated in Figure \ref{fig:grids}. 

\section{The main results}
\label{sec.main}

Our main results are the following Large Deviations and Central Limit theorems. For reader's convenience we outline the key steps of the proofs below.  

\begin{thm}\label{th.main}
Suppose that the initial state $\rho$ is a finite rank operator with respect to the Fock basis, or more generally, that
$\sum_{n \geq 0} |\langle e_n, \rho e_n \rangle|^2\rss(n)^{-1}<+\infty$. Then the counting process $\Lambda_t$ satisfies the large deviations principle with rate function equal to the Legendre transform of $\lambda(s)$, where $\lambda(s)$ is the limit in \eqref{eq:mfgconvergence}. The function $\lambda(s)$ is smooth and it is equal to the spectral bound of a certain semigroup generator $L^{(d)}_s$ defined below. 
\end{thm}
In particular the atom maser does not exhibit dynamical phase transitions, but rather cross-over transitions which become sharper as $N_{ex}$ increases. This behaviour will be investigated in more detail in section \ref{sec.numerics}.

 \begin{col}\label{cor:main}
 The counting process $\Lambda_{t}$ satisfies the Central Limit Theorem 
 $$
 \frac{1}{\sqrt{t}} ( \Lambda_t - t \cdot m) \overset{\mathcal{D}}{\longrightarrow} N(0, V), 
 $$
 where $\mathcal{D}$ denotes convergence in distribution and $m$ and $V$ are the mean and variance
 $$
 m = \left. \frac{d \lambda(s) }{ds}\right|_{s=0}=\frac{\mathbb{E}_{\rss}(\Lambda_t)}{t}, \quad V=\left. \frac{d^2 \lambda(s) }{ds^2}\right|_{s=0}.
 $$

 \end{col}
 
 \noindent
 \emph{Proof of Theorem \ref{th.main} and Corollary \ref{cor:main}}. 
 
 \noindent
 For clarity of the exposition we break the proof into individual steps.
 
 \noindent
(i) We introduce $L^2(\rss)$ as the completion of $B(\mathfrak{h})$ endowed with the norm 
$\|\cdot\|_2$ induced by the following inner product
\begin{equation}\label{eq:eL2}
\langle Y,X\rangle=\Tr((\rss^{1/4}Y\rss^{1/4})^*(\rss^{1/4}X\rss^{1/4})).
\end{equation}
We recall (see \cite{Carbone2000}, especially Proposition 2.1) that $L^2(\rss)$ is isomorphic as an Hilbert space to the Schatten ideal $L^2(\Tr):=\{X \in B(\mathfrak{h}): \Tr(X^*X)<+\infty\}$ via the unique continuous extension of the correspondence
\[\begin{split}
    i:B(\mathfrak{h}) & \rightarrow L^2(\Tr)\\
    X &\mapsto \rss^{1/4}X\rss^{1/4}.
\end{split}
\]
Hence, with an abuse of notation, we will identify $Y \in L^2(\rss)$ with the corresponding operator in $L^2(\Tr)$.
We recall that for every $X_n,X \in B(\mathfrak{h})$ and $Y \in L^2(\rss)$
\begin{enumerate}
\item $\|X\|_2 \leq \|X\|_\infty$ (\cite[Proposition 2.1]{Carbone2000}) and
\item if $X_n \xrightarrow{w^*} X$, then $\langle Y,X_n \rangle \rightarrow \langle Y, X\rangle$: indeed
\[\begin{split}
    \lim_{n \rightarrow +\infty} \langle Y,X_n \rangle&= \lim_{n \rightarrow +\infty} \Tr( Y^* \rss^{1/4}X_n\rss^{1/4})=\lim_{n \rightarrow +\infty} \Tr( \rss^{1/4}Y^* \rss^{1/4}X_n)= \\
    &=\Tr( \rss^{1/4}Y^* \rss^{1/4}X)= \langle Y, X \rangle.
\end{split}
\]
\end{enumerate}
\begin{lem} \label{lem:l2ext}
For every $s \in \mathbb{R}$ the following hold
\begin{enumerate}[1.]
\item 
There exists a unique strongly continuous semigroup $(T_s(t))_{t \geq 0}$ of bounded linear maps on $L^{2}(\rss)$ such that
\begin{equation*}
	T_s(t)(X) =\mathcal{T}_s(t)(X),\quad X \in \mathcal{B}(\mathfrak{h}).
\end{equation*}
Every $X \in D(\mathcal{L}_s)$ belongs to the domain of the generator $L_s$ of $(T_s(t))_{t \geq 0}$ and
\begin{equation*}
	L_s(X) =\mathcal{L}_s(X),\quad X \in D(\mathcal{L}_s)=D(\mathcal{L}).
\end{equation*}
As usual we denote $(T(t))_{t \geq 0}$ and $L$ the semigroup and the corresponding generator in the case $s=0$.
\item The set $\mathcal{M}(\mathfrak{h})$ of finite rank operators given by finite matrices with respect to the Fock basis forms a core for $L_s$.
\item $L_s = L+ \delta_s$, with $\delta_s$ a bounded perturbation. 
\end{enumerate}
\end{lem}

\emph{Proof of Lemma \ref{lem:l2ext}}

1. First, we show that $\rss$ is sub-invariant for ${\cal T}^\prime$ (as defined in the proof of points 1. and 2. of Proposition \ref{prop:defdef}). Since $\rss \in D({\cal L}^\prime_*)$ and $\rss$ is invariant for ${\cal T}$, this is equivalent to showing that
$$
{\cal L}^\prime_*(\rss)=({\cal L}^\prime_*-{\cal L}_*)(\rss)=(e^s-1) \left ({\cal J}_{1*}(\rss)-1_{s>0}N_{ex}\left ( 1+ \frac{1}{\nu} \right) \rss\right ) \le 0.
$$
This inequality is trivial for $s<0$, while, for $s>0$, we can write
 the explicit form of $\rss$ and get
\[\begin{split}
{\cal J}_{1*}(\rss)&=N_{ex}\sum_{n \geq 0}  \sin^2(\phi \sqrt{n+1}) \rss(n) \ket{n+1}\bra{n+1}=\\
&=N_{ex}\sum_{n \geq 0}  \sin^2(\phi \sqrt{n+1}) \frac{\rss(n)}{\rss(n+1)}\rss(n+1)\ket{n+1}\bra{n+1}=\\
&\leq N_{ex}\left ( 1+ \frac{1}{\nu} \right) \rss
\end{split}
\]
since
\[ \frac{\rho_{ss}(n)}{\rho_{ss}(n+1)}=\frac{(\nu+1)(n+1)}{\nu(n+1)+N_{ex}\sin^2(\phi\sqrt{n+1})}\leq 1+\frac{1}{\nu}.
\]
Then we apply \cite[Theorem 2.3]{Carbone2000} and we obtain that the semigroup ${\cal T}^\prime$ can be extended to a strongly continuous contraction semigroup on $L^2(\rss)$. Also here, the conclusion follows multiplying the semigroup by a suitable exponential factor as for points 1. and 2. of Proposition \ref{prop:defdef}.

\sloppy 2. and 3. The proof of the previous point shows that $ \left (N_{ex}\left(1+ \frac{1}{\nu} \right ) \right )^{-1} {\cal J}_{1*}(\rss) \leq \rss$, hence we can apply again \cite[Theorem 2.3]{Carbone2000} in order to show that $\left (N_{ex}\left(1+ \frac{1}{\nu} \right ) \right )^{-1} {\cal J}_1$ extends to a bounded operator on $L^2(\rss)$ and so does ${\cal L}_s-{\cal L}=(e^s-1){\cal J}_1$; let us call $\delta_s$ such extension. Since on ${\cal M}(\mathfrak{h})$ we have $L_s=L+\delta_s$, it follows that if we prove that ${\cal M}(\mathfrak{h}) \subset D({\cal L})$ is a core for $L_s$, we have that $D(L_s)=D(L)$ and $L_s=L+\delta_s$. Notice that both $L$ and $\delta_s$ preserve ${\cal M}(\mathfrak{h})$, that ${\cal M}(\mathfrak{h})$ is dense in $B(\mathfrak{h})$ in the $w^*$-topology and hence in $L^2(\rss)$ in norm. Hence we only need to show that ${\cal M}(\mathfrak{h})$ is a set of analytic vectors for $L_s$ and then apply Proposition \ref{prop:core} in Appendix \ref{appendix:other.restults}. 

Let us fix $s \in \mathbb{R}$. Then, for $n,m \in \mathbb{N}$ we have
\be \label{eq:Lenem}
L_s(\ket{e_n}\bra{e_m})=\alpha_{m,n} \ket{e_{n-1}}\bra{e_{m-1}}+\beta_{m,n} \ket{e_n}\bra{e_m}+\gamma_{m,n}\ket{e_{n+1}}\bra{e_{m+1}}
\ee
where

\begin{align*}
&\alpha_{m,n}=\nu \sqrt{nm}+N_{ex}e^s\sin(\phi\sqrt{n})\sin(\phi\sqrt{m}),\\
&\beta_{m,n}=-\frac{1}{2}\left ((\nu+1)(m+n) + \nu (m+n+2) \right )+N_{ex}( \cos(\phi\sqrt{n+1})\cos(\phi\sqrt{m+1})-1),\\
&\gamma_{m,n}=(\nu+1)\sqrt{n+1}\sqrt{m+1}.\\
\end{align*}
Notice that $|\alpha_{m,n}|+|\beta_{m,n}|+|\gamma_{m,n}|\leq \underbrace{3(\nu+1+2(e^s \vee 1)N_{ex})}_{=:B}(n+m+1)$, hence
\[\begin{split}
\lVert L_s^k(\ket{e_n}\bra{e_m} \rVert_{\infty} &\leq B^k (n+m+1)\cdot (n+m+3) \cdots (n+m+2k-1) \leq\\
& \leq B^k 2^k (n+m+1) (n+m+2) \cdots (n+m+k) =\\
&=B^k 2^k k! \begin{pmatrix} n+m+k \\ n+m \end{pmatrix} \leq 2^{n+m}4^k B^k k!.\\
\end{split}
\]
If $\overline{T}=1/(4B)$, for any $t<\overline{T}$, $\sum_{k=0}^{+\infty} \frac{t^k L_s^k(\ket{e_n}\bra{e_m})}{k!}$ converges uniformly on compact intervals in the uniform norm, hence with respect to $\norm{\text{ }}_2$. By the definition of ${\cal M}(\mathfrak{h})$, we can deduce the same for every $X \in {\cal M}(\mathfrak{h})$. Let us call $X_t:=\sum_{k=0}^{+\infty} \frac{t^k L_s^k(X)}{k!}$ and notice that, since the series of the derivatives converges uniformly on compact intervals, it solves the following abstract Cauchy problem:
\[X_t=X+L_s\int_0^t X_u du, \quad t\leq \overline{T}.
\]
Hence, by \cite[Proposition 6.4]{Engel2006}, $X_t=T_s(t)(X)$ for every $t<\overline{T}$.

\qed

 \noindent
(ii)
The moment generating function of $\Lambda_t$ for an initial state $\rho$ (cf. equation \eqref{eq:mgf}) can be expressed in terms of the semigroup acting on $L^2(\rss)$ as 
	\begin{equation*}
		\mathbb{E}_{\rho}(e^{s \Lambda_{t}}) = 
		\Tr ( \rho \mathcal{T}_{s}(t)(\mathbf{1}))
		=\langle \tilde{\rho}, T_{s}(t)(\mathbf{1}) \rangle,
	\end{equation*}
\sloppy where $\tilde{\rho}:= \rho_{ss}^{-1/2} \rho \rho_{ss}^{-1/2}$ is assumed to belong to $L^2(\rss)$, or equivalently $\sum_{n \geq 0} |\langle e_n, \rho e_n \rangle|^2\rss(n)^{-1}<+\infty$ (in this case $\tilde{\rho}$ extends to a bounded linear functional on $L^2(\rss)$). This holds for instance if $\rho$ has a finite number of photons, but also if $\rho= \rho_{ss}$.

 \noindent
(iii) As we already mentioned, the commutative Von Neumann algebra $B_d(\mathfrak{h})$ of the operators which are diagonal in the Fock basis plays a fundamental role. We need to introduce some other related linear spaces:
\[
\begin{split}
 & {\cal M}_d(\mathfrak{h}):={\cal M}(\mathfrak{h}) \cap B_d(\mathfrak{h})={\rm span}\{ \ket{e_n}\bra{e_n}: n \in \mathbb{N}\}\subset D(\mathcal{L}),\\
 & L^2_d(\rss):=\overline{{\cal M}_d(\mathfrak{h})}^{\| \text{ } \|_2}.\\
 \end{split}
\]
Notice that $B_d(\mathfrak{h}) \simeq \ell^\infty(\mathbb{N})$ as Von Neumann algebras and $L^2_d(\rss) \simeq \ell^2(\mathbb{N},\rss)$ as Banach spaces.
\begin{prop} \label{prop:diag}
The following statements hold for every $s \in \mathbb{R}$.

1. The generator $\mathcal{L}_s$ and the corresponding semigroup $(\mathcal{T}_s(t))_{t\geq 0}$ preserve the algebra $B_d(\mathfrak{h})$. Consequently also $L_s$ and $T_s$ preserve the subspace $L^2_d(\rss)$.

2. The action of $\mathcal{L}_s$ on the diagonal is explicitly written as an operator $\mathcal{L}_s^{(d)}$ acting on $f=\sum_k f_k \ket{e_k}\bra{e_k} \in D(\mathcal{L}_s^{(d)}) :=D(\mathcal{L}_s) \cap B_d(\mathfrak{h})$  as $\mathcal{L}_s^{(d)}(f)=
g=\sum_k g_k \ket{e_k}\bra{e_k}
$
with
\begin{eqnarray}\label{diagonal generator}
g_k &=&   k(\nu+1)(f_k-f_{k-1})
+ (e^s-1) N_{ex}\sin^2(\phi\sqrt{k+1}) f_{k+1}
\nonumber\\
&& +(\nu(k+1)+N_{ex}\sin^2(\phi\sqrt{k+1}))(f_{k+1}-f_k).
\end{eqnarray}
In particular, $\mathcal{L}_0^{(d)}=\mathcal{L}^{(d)}$ is the generator of a birth-death process with birth and death rates as in equation \eqref{eqn:rates}.
\end{prop}

\emph{Proof of Proposition \ref{prop:diag}}

1. $\mathcal{M}_d(\mathfrak{h}) \subset D(\mathcal{L}_s) \cap B_d(\mathfrak{h})$ is $w^*$-dense in $B_d(\mathfrak{h})$, hence it is sufficient to compute explicitly $\mathcal{L}_s(\ket{e_k}\bra{e_k})$ and observe that it belongs to 
$${\rm span}\{\ket{e_j}\bra{e_j}, j=k-1,k,k+1\}$$ for all $k$. Then $\mathcal{L}_s({\mathcal M}_d(\mathfrak{h}))\subseteq {\mathcal M}_d(\mathfrak{h})$. The rest follows from the definitions.
\\
2. This point is a direct computation by using the definition of the generator $\mathcal{L}_s$ given by equations \eqref{eq:lindblad} and \eqref{eq:perturbed}.

\qed


\noindent
We denote by $L_s^{(d)}$ and $T_s^{(d)}$ the restrictions of $L_s$ and $T_s$ to $L^2_d(\rss)$ (as usual we drop the index $s$ in the case $s=0$). Since $\mathbf{1} \in L^2_d(\rss) $, the moment generating function can be expressed as 
$$
\mathbb{E}_{\rho}(e^{s \Lambda_{t}}) = \langle \tilde{\rho}^{(d)}, T^{(d)}_{s}(t)(\mathbf{1}) \rangle
$$
with $ \tilde{\rho}^{(d)}$ denoting the diagonal of  $ \tilde{\rho}$.


 \noindent
(iv) The semigroup $(T^{(d)}_s(t))_{t \geq 0}$ is \emph{immediately compact} (i.e. $T^{(d)}_s(t)$ is compact for all $t>0$). In order to prove it we first show it for $T^{(d)}(t)$ with the classical theory of birth-death processes and then use perturbation theory.

\begin{lem}\label{lemma:selfadj.gen}
The following statements hold.

\begin{enumerate}[1.]
\item $L^{(d)}$ is selfadjoint and its essential spectrum is empty,
\item $\left(T^{(d)}(t)\right)_{t \geq 0}$ is immediately compact,
\item $\left(T^{(d)}_s(t)\right)_{t \geq 0}$ is immediately compact for all $s\in \mathbb{R}$.
\end{enumerate}
\end{lem}
We recall that $\sigma(L^{(d)})$ is the disjoint union of the discrete and the essential spectrum and the discrete spectrum for a selfadjoint operator is defined as those $\alpha \in \mathbb{C}$ which are isolated eigenvalues with finite multiplicity (see \cite[Theorem VII.10]{Reed1972}).

\emph{Proof of Lemma \ref{lemma:selfadj.gen}}

1. The proof of point 2. of Lemma \ref{lem:l2ext} shows that ${\cal M}_d(\mathfrak{h})$ is a dense linear space of analytic vectors for $L^{(d)}$, hence, by Nelson theorem (\cite[Theorem X.39]{Reed1975}), in order to establish the selfadjointness of $L^{(d)}$, it is enough to check that it is symmetric on ${\cal M}_d(\mathfrak{h})$ (which is a simple computation using equation \eqref{eq:Lenem}).

In order to show that the essential spectrum of $L^{(d)}$ is empty it suffices to check that the following condition on birth and death rates holds (cf. \cite[Theorem 1.2]{Mao2006} and we already know that the process is non-explosive and admits a unique invariant state):
\[\lim_{n \rightarrow +\infty} \left (\sum_{i=1}^{n} \frac{1}{\mu^2_i \rss(i)} \right ) \cdot \sum_{j=n+1}^{+\infty} \rss(j)=0.
\]
Let us do the computations:
\[\begin{split}
&\sum_{i=1}^{n} \frac{1}{\mu^2_i\rss(i)} \cdot \sum_{j=n+1}^{+\infty} \rss(j)=\sum_{i=1}^{n}\sum_{j=n+1}^{+\infty} \frac{1}{\mu^2_i} \frac{\rss(j)}{\rss(i)}\\
&=\sum_{i=1}^{n}\sum_{j=n+1}^{+\infty} \frac{1}{\mu^2_i} \prod_{k=i+1}^{j} \frac{\nu + N_{ex} \sin^2(\phi\sqrt{k})/k}{\nu+1}.\\
\end{split}
\]
Notice that for every $0 <\epsilon<1$, there exists $M \in \mathbb{N}$ such that for every $k \geq M$, $N_{ex} \sin^2(\phi\sqrt{k}) \leq   \epsilon k$, therefore for $n > M$ we have
\[\begin{split}
&\sum_{i=1}^{n} \frac{1}{\mu^2_i\rss(i)} \cdot \sum_{j=n+1}^{+\infty} \rss(j)\leq \underbrace{\sum_{i=1}^{M} \frac{1}{\mu^2_i\rss(i)}\cdot \sum_{j=n+1}^{+\infty} \rss(j)}_{=o(1)}\\
&+  \sum_{i=M+1}^{n} \frac{1}{\mu^2_i} \cdot \sum_{j=n+1}^{+\infty} \left (\underbrace{ \frac{\nu +\epsilon}{\nu+1}}_{=:q} \right )^{j-i-1}\leq  \frac{1}{1-q}\sum_{i=M+1}^{n} \frac{q^{n-i}}{(\nu+1) i}+o(1)\\
&= \frac{1}{1-q}\sum_{i=0}^{n-M-1} \frac{q^i}{(\nu+1) (n-i)}+o(1).\\
\end{split}
\]
The result follows from the fact that $q<1$ and the dominated convergence theorem.\\
2. $L^{(d)}$ is a selfadjoint, unbounded operator with empty essential spectrum and such that $L^{(d)}\leq 0$, therefore its spectral resolution reads
\[L^{(d)}=- \sum_{n \geq 0} \alpha_n P_n,
\]
where $P_n$'s are finite dimensional orthogonal projections and $\alpha_n$'s are distinct non-negative real numbers such that $\lim_{n \rightarrow +\infty}\alpha_n=+\infty$ (they do not accumulate). The semigroup generated by $L^{(d)}$ can be expressed via functional calculus as
\[T^{(d)}(t)=\sum_{n \geq 0} e^{-t\alpha_n} P_n, \quad t \geq 0.
\]
By its spectral representation, we can conclude that $T^{(d)}$ is immediately compact.
\\
3. The restriction $L^{(d)}_{s}$ is a bounded perturbation of the generator $L^{(d)}$, hence the semigroup $\left(T^{(d)}_s(t)\right)_{t \geq 0}$ is also immediately compact, cf. \cite{Engel2001} Thm. III.1.16. 

\qed

Since $(T^{(d)}_{s}(t))$ is an immediately compact semigroup, we have (\cite{Engel2001}, Col. IV.3.12) a spectral mapping theorem of the form
\begin{equation*}
	e^{t \sigma(L^{(d)}_{s})} = \sigma(T^{(d)}_{s}(t)) \setminus \set{0},\quad t > 0;
\end{equation*}
in particular, the \emph{spectral radius} of $T^{(d)}_{s}(t)$ is given by
\begin{equation*}
r_s(t):=	r(T^{(d)}_{s}(t)) = e^{t \lambda(s)}.
\end{equation*}
where $\lambda(s)$ is the \emph{spectral bound} of $L^{(d)}_{s}$, i.e. the real part of the eigenvalue with the largest real part.

 \noindent
(v) The semigroup $(T^{(d)}_{s}(t))_{t \geq 0}$ is strictly positive, that is $T^{(d)}_{s}(t)(D)>0$ for all 
$D\geq 0$ in $L^2_d(\rss)$ and $t> 0$.
\begin{proof}
It is not difficult to see that every $D\geq 0$ in $L^2_d(\rss)$ is of the form $D=\sum_k D_k \ket{e_k}\bra{e_k}$ for some $D_k \in \ell^2(\mathbb{N}, \rss)$ and $D_k\geq 0$ for every $k$, hence it is enough to show that for every $k$ $T_s^{(d)}(\ket{e_k}\bra{e_k})>0$ for every $t >0$. Notice that
$(P_{m,n}(t):={\rm e}^{t (1-e^s)\left ( 1+ \frac{1}{\nu} \right )N_{ex}1_{\{s>0\}}}\langle e_n,T_s^{(d)}(t)(\ket{e_m}\bra{e_m}) e_n \rangle)_{n,m \in \mathbb{N}}$ is a standard transition function and so by well known results of continuous time Markov chains we have that either $P_{m,n}$ is constant and equal to $0$ or it is strictly positive for every $t>0$ (Levy's Theorem, \cite[Proposition 1.3]{An91}). Since the birth and death rates are all strictly positive, $P_{m,n}(t)>0$ for every positive time (for a reference about continuous time Markov chains see for instance \cite{An91}).
\end{proof}


(vi) Since $\left(T^{(d)}_{s}(t)\right)_{t > 0}$ is compact and strictly positive, Krein-Rutman theory (\cite[Theorem 1.5]{Chang2020}) implies that the spectral radius of $T^{(d)}_{s}(t)$ is an algebraically simple eigenvalue with strictly positive right and left eigenvectors $r(s)$ and $l(s)$.
\begin{lem}\label{lemma:dominant.term}
For every $s \in \mathbb{R}$, there exists a strictly positive number $g(s)>0$ such that for every $t>0$ 
\be \label{eq:spgap}
T_s^{(d)}(t)= e^{t\lambda(s)} (\ket{r(s)}\bra{l(s)} + R_t^s)
\ee
and $\| R_t^s \|_{L^2_d(\rss)\rightarrow L^2_d(\rss)} , =O(e^{-g(s)t})$.
\end{lem}

\emph{Proof of Lemma \ref{lemma:dominant.term}}

Because of compactness of $T_s^{(d)}(t)$, we only need to prove that
$$\sigma(T_s^{(d)}(t)) \cap \{z \in \mathbb{C}:|z|=e^{\lambda(s)t}\}=\{e^{\lambda(s)t}\}.$$
Suppose that there exist $\theta \in [0,2\pi)$ and $D \in L_d^2(\rss)$ such that $L_s^{(d)}(D)=(\lambda(s)+i\theta)D$, then if we consider $t=2\pi/\theta$, $T_s^{(d)}(t)(D)=e^{\lambda(s)t}D$. Since $l(s)$ is the unique eigenvector with eigenvalue one, we conclude that $\theta=0$.

\qed

Using point (iii) 
this implies that 
\begin{eqnarray}
\mathbb{E}_{\rho} (e^{s \Lambda_t} ) &=& e^{t\lambda(s)} \left( \langle    \tilde{\rho}, r(s) \rangle \langle   l(s), \mathbf{1} \rangle  + o(1)\right).
\label{eq.dominant.term}
\end{eqnarray}
Since $l(s),r(s) >0$ and $\tilde{\rho}, \mathbf{1} \geq 0$ the inner products are non-zero and we obtain the limiting cumulant generating function
$$
\lim_{t\to\infty} \frac{1}{t} \log \mathbb{E}_{\rho_{\rm in}} (e^{s \Lambda_t} ) = \lambda(s).
$$

 \noindent
(vii)
Using analytic perturbation theory for the generator $L^{(d)}_{s}$, the spectral bound $\lambda(s)$ can be shown to be a smooth function of $s$ in a complex neighborhood of the real line. Since for every $s_0 \in \mathbb{R}$, $\lambda(s_0)$ is an isolated eigenvalue with finite multiplicity, we can apply \cite[Prop. 3.25, p. 141]{Chatelin1983} to the family of perturbations
	\begin{equation*}
		V_{s-s_0}:=L^{(d)}_{s} = L_{s_0}^{(d)} + \delta_{s}-\delta_{s_0}=L_{s_0}^{(d)} + e^{s_0}(e^{s-s_0}-1) J_1^{(d)}
	\end{equation*}
and we find that $\lambda(s)$ is an analytic function of $s$ and remains isolated in a complex neighborhood of $s_0$.

 \noindent
(viii)
Using points (vi) 
and (vii), 
 we apply G\"{a}rtner-Ellis theorem to conclude that $\Lambda_t$ satisfies the LD principle with rate function equal to the Legendre transform of $\lambda(s)$.

 \noindent
(ix)
Furthermore, by the result of \cite{Bryc}, it follows that $\Lambda_t$ satisfies the CLT. In particular, the limiting cumulants of $\Lambda_t$ can be computed as derivatives of 
$\lambda(s)$ at $s=0$,

$$
\lim_{t\to\infty} \frac{1}{t} C_k(\Lambda_t) = \left. \frac{d^k \lambda(s)}{ds^k}\right|_{s=0}.
$$
By $J_1^{(d)}$ we denote the operator acting on $D=\sum_k D_k \ket{e_k}\bra{e_k} \in L^2_d(\rss)$ in the following way:
\[J_1^{(d)}(D)=\sum_{k \geq 1} N_{\rm ex} \sin^2(\phi \sqrt{k}) D_{k} \ket{e_{k-1}}\bra{e_{k-1}}.
\]
\begin{lem}\label{lemma:derivatites.lambda}
\begin{align*}
\left. \frac{d\lambda(s)}{ds}\right|_{s=0}&=\langle \mathbf{1}, J^{(d)}_1(\mathbf{1}) \rangle,\\
\left. \frac{d^2 \lambda(s)}{ds^2}\right|_{s=0}&=\langle \mathbf{1}, J^{(d)}_1(\mathbf{1}) \rangle + 2 \langle \mathbf{1}, J^{(d)}_1(D_V) \rangle
\end{align*}
and $D_V$ can be characterized as the unique solution in $L^2_d(\rss)$ of
\[L^{(d)}(D_V)=\langle \mathbf{1}, J^{(d)}_1(\mathbf{1}) \rangle\mathbf{1}-J^{(d)}_1(\mathbf{1}).
\]
\end{lem}

\emph{Proof of Lemma \ref{lemma:derivatites.lambda}}

Differentiating first once and then twice
\[L_s^{(d)}(r(s))=\lambda(s)r(s)
\]
and evaluating in $s=0$ we get
\begin{equation} \label{eq:mu}
L^{(d)}(r^\prime(0))+J^{(d)}_1(\mathbf{1})=\lambda^\prime(0) \mathbf{1}
\end{equation}
\begin{equation} \label{eq:sigma}
L^{(d)}(r^{\prime \prime}(0))+J^{(d)}_1(\mathbf{1})+2 J^{(d)}_1(r^\prime(0))=\lambda^{\prime\prime}(0) \mathbf{1}+2\lambda^\prime(0) r^\prime(0)
\end{equation}
where we used the fact that $\lambda(0)=0$ and $r(0)=\mathbf{1}$. Notice that $l(0)=\mathbf{1}$ too, hence taking the scalar product of equation \eqref{eq:mu} against $\mathbf{1}$ we get
\[\lambda^\prime(0)=\langle \mathbf{1}, J^{(d)}_1(\mathbf{1}) \rangle.
\]
Substituting the expression we obtained for $\lambda^\prime(0)$ in Equation \eqref{eq:mu} we get that $r^\prime(0)$ is the unique (remember that $\ker(L^{(d)})$ has dimension $1$) solution of
\[L^{(d)}(r^\prime(0))=\langle \mathbf{1}, J^{(d)}_1(\mathbf{1}) \rangle \mathbf{1}-J^{(d)}_1(\mathbf{1}).
\]
We can choose $r(s)$ such that $\langle \mathbf{1}, r(s) \rangle =1$, hence substituting the expression we obtained for $\lambda^\prime(0)$ in equation \eqref{eq:sigma} and taking the scalar product against $\mathbf{1}$ we get
\[\lambda^{\prime \prime}(0)=\langle \mathbf{1}, J^{(d)}_1(\mathbf{1}) \rangle + 2 \langle \mathbf{1}, J^{(d)}_1(r^\prime(0)) \rangle.
\]

\qed

Notice that $\lambda^\prime(0)$ is the expected value of $\frac{\Lambda_t}{t}$ if the system starts in the stationary state $\rss$ (Equation (\ref{eq.statmean})), indeed:
\[
\langle \mathbf{1}, J^{(d)}_1(\mathbf{1}) \rangle= \Tr(\rss L_1^*L_1)=\sum_{n\geq 0}  n\rss(n)-\nu.
\]

\section{Numerical analysis}\label{sec.numerics}

\begin{figure}[t]
    \centering
    \begin{subfigure}[b]{0.45\textwidth}
        \centering
        \includegraphics[width=\textwidth]{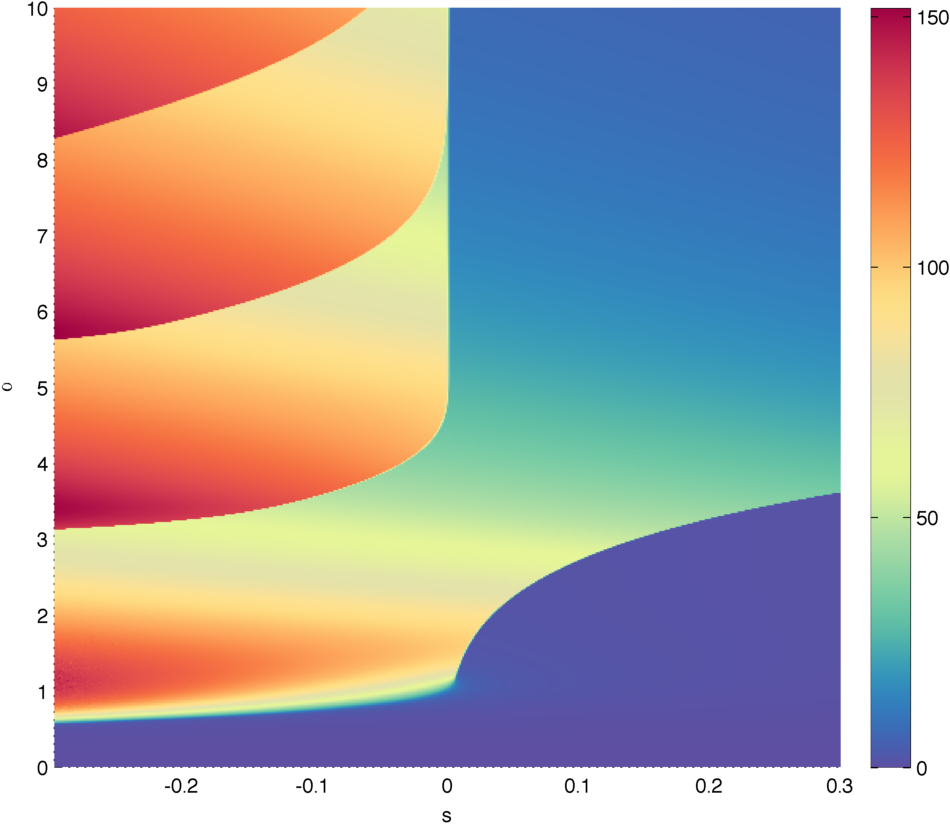}
        \caption{$\lambda'(s)$}
    \end{subfigure}
    \quad
    \begin{subfigure}[b]{0.45\textwidth}
        \centering
        \includegraphics[width=0.975\textwidth]{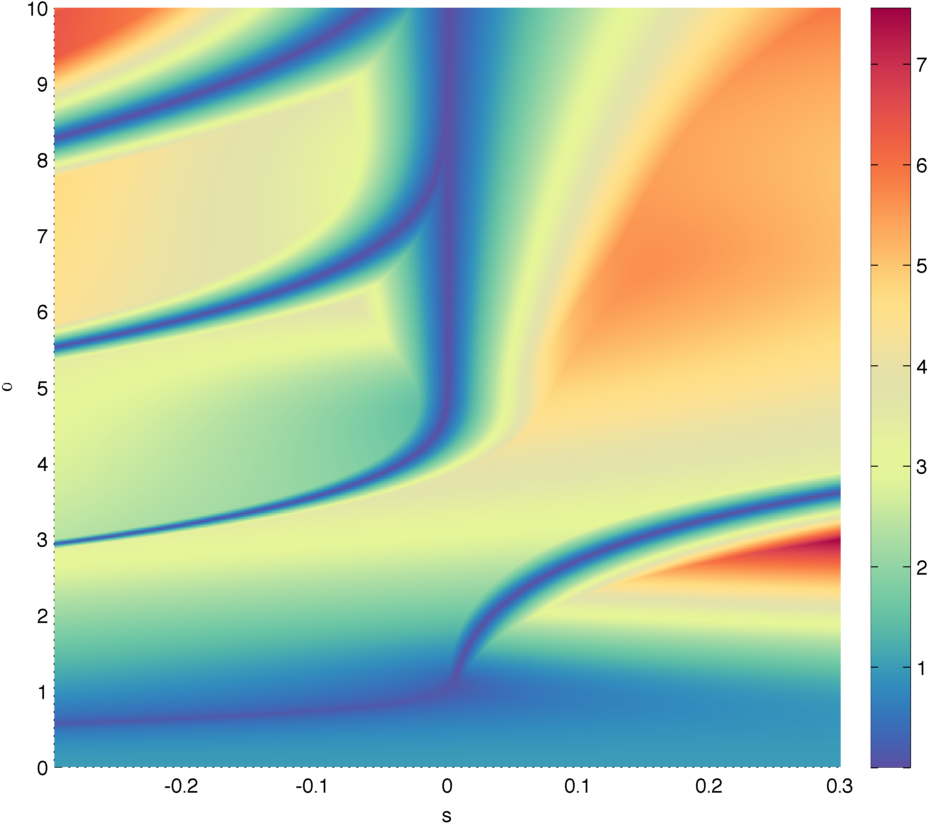}
        \caption{g(s)}
    \end{subfigure}
    \caption{Derivative $\lambda'(s)$ (panel (a)) and the spectral gap $g(s)$ of $L_{s}$ (panel (b)) as functions of $s$ and $\alpha = \phi /\sqrt{N_{\text{ex}}}$ (after Fig. 3 in \cite{Garrahan2011}).}\label{fig:grids}
\end{figure}

The existence of a  "phase transition" in the atom maser has  been discussed in several theoretical physics papers \cite{Briegel1994,Englert2002a,Benson1994,Rempe1990,Garrahan2011}. There is a general agreement that if $N_{ex}$ is sufficiently large (for instance $N_{ex}\approx 150$ ), then "for all practical purposes" we can consider that the mean photon number of the stationary state has a jump at $\alpha\approx 6.66$ (see Figure \ref{fig:stationary}) which matches up  with a jump between the left and right derivatives of $\lambda(s)$ at $s=0$, in the dynamical scenario (see Figure \ref{fig:grids}). However, the question whether we are dealing with a "true" (dynamical) phase transition or rather a steep but smooth cross-over was left open, and motivated this investigation. Having proved that the latter is the case, we would like to briefly put the result in the context of a numerical analysis.

As the proof suggests, dynamical phase transitions are intimately connected with the closing of the spectral gap of the semigroup generator. Figure \ref{fig:grids} shows the close match between the behaviour of the first derivative of 
$\lambda(s)$ and the spectral gap $g(s):=\lambda(s)- {\rm Re} \lambda_1(s)$. In particular, at first sight it would appear that for $\alpha\geq 4.6$ (the point where the stationary state becomes bistable), the entire $s=0$ line is a phase separation line. However, by zooming in a vertical strip of size $10^{-7}$ in this region (see Figure \ref{fig:phaseboundaries}), we find that the line separating the phases is not perfectly vertical but crosses $s=0$ at 
$\alpha\approx 6.6$ which corresponds roughly to the transition point for the stationary state. Moreover, on this scale it is clear that we deal with a steep but smooth transition between phases.

\begin{figure}[t]
    \centering
        \begin{subfigure}[b]{0.95\textwidth}
          \centering
          \begin{subfigure}[b]{0.31\textwidth}
                  \centering
                  \includegraphics[width=\textwidth]{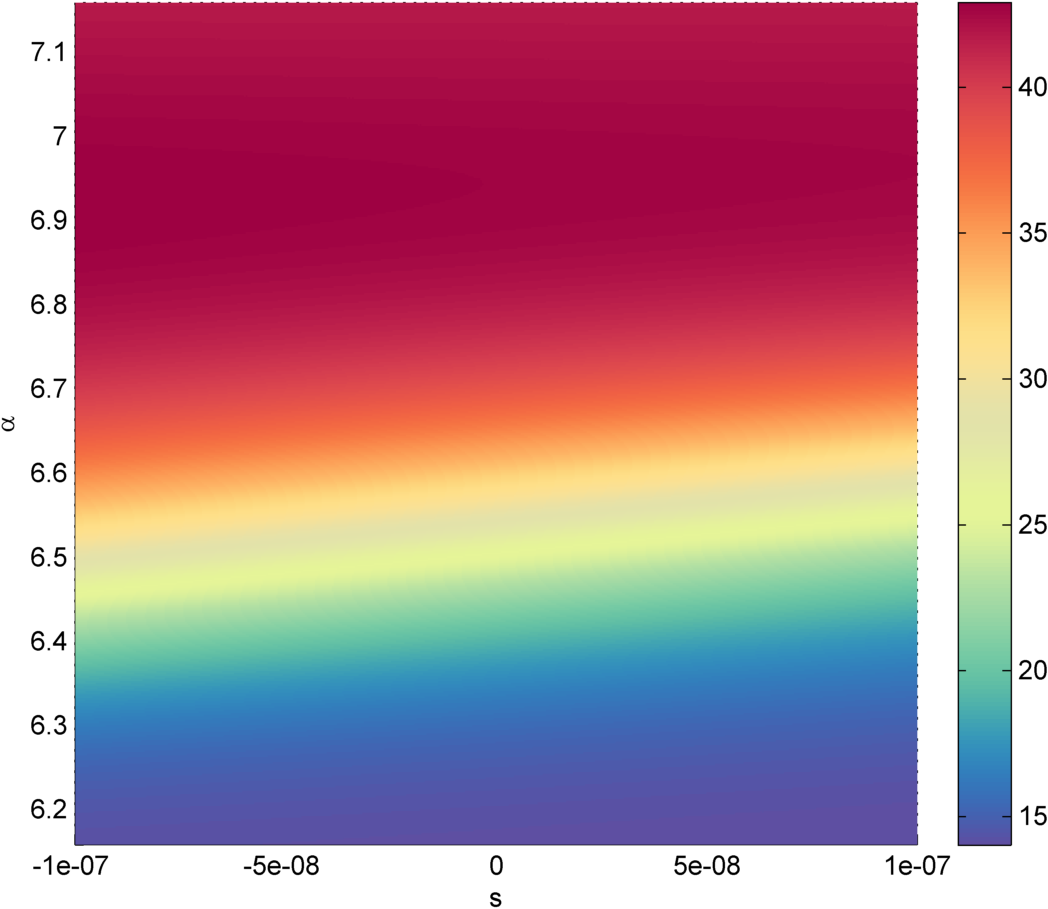}
          \end{subfigure}%
          \begin{subfigure}[b]{0.31\textwidth}
                  \centering
                  \includegraphics[width=\textwidth]{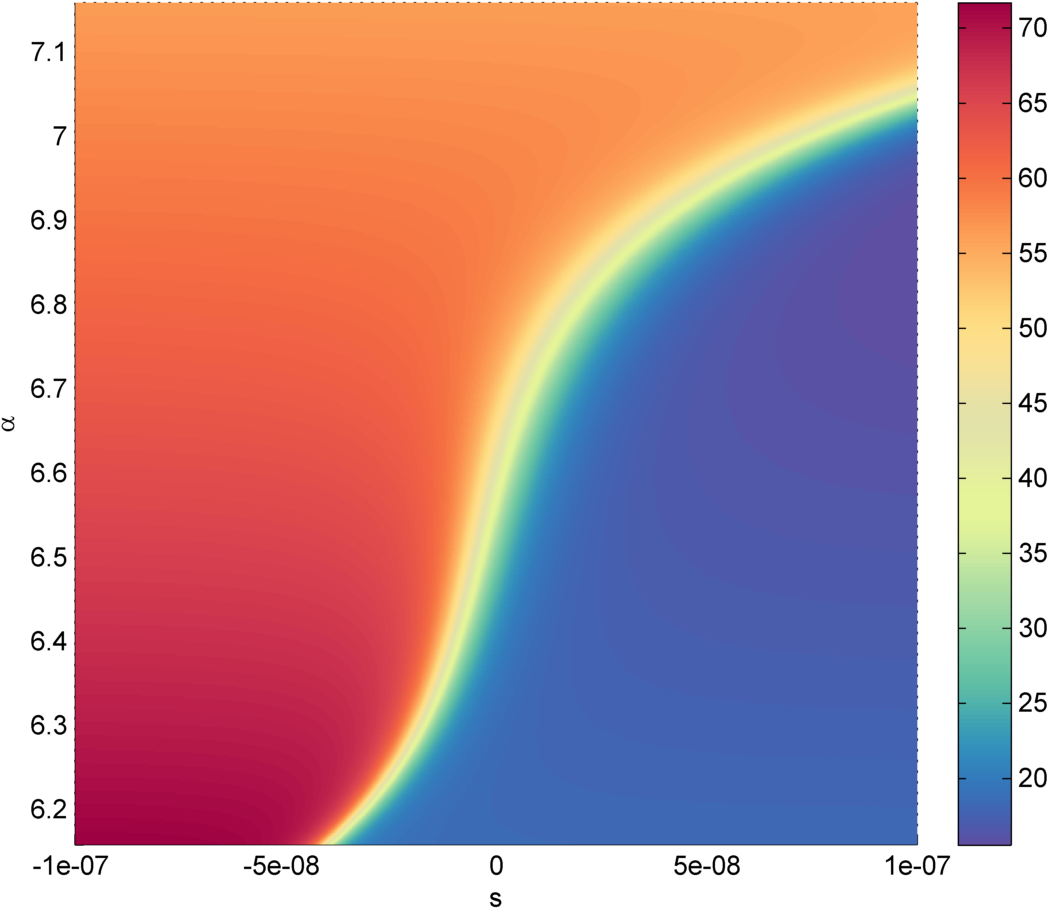}
          \end{subfigure}
                    \begin{subfigure}[b]{0.31\textwidth}
                  \centering
                  \includegraphics[width=\textwidth]{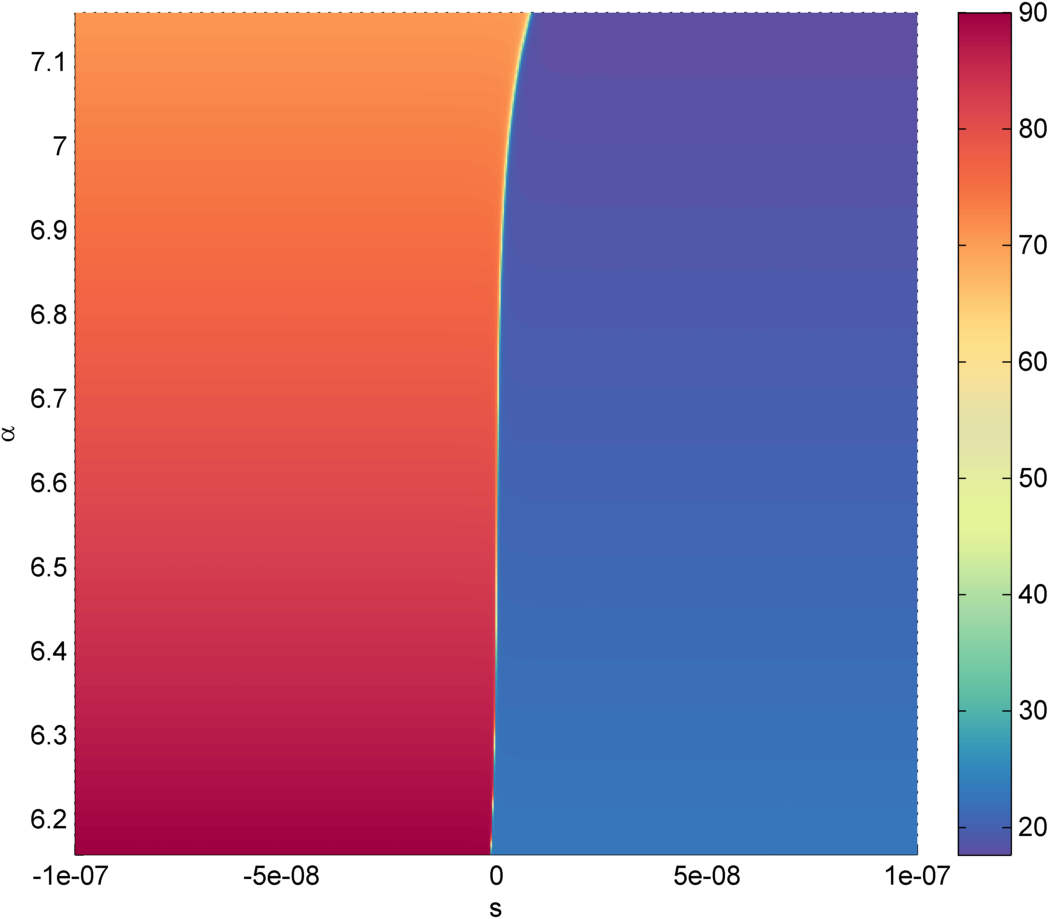}
          \end{subfigure}
    \end{subfigure}
    \caption{Phase boundaries at the $s=0, \alpha \approx 6.66$ crossover with $N_{ex} = 75,100$ and $125$. 
    }\label{fig:phaseboundaries}
\end{figure}

\begin{figure}[t]
    \centering
    \begin{subfigure}[b]{0.65\textwidth}
    \centering
    \includegraphics[width=0.95\textwidth]{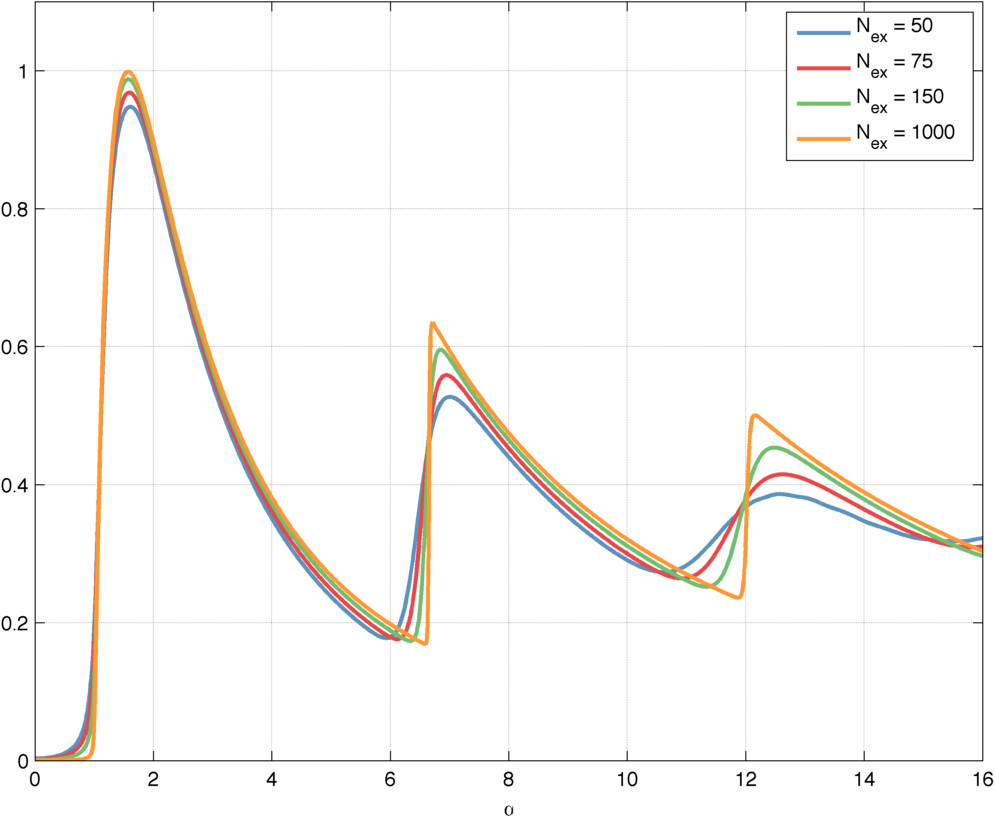}
    \end{subfigure}
    \caption{Rescaled stationary state mean photon numbers, $\langle \rho_{\text{ss}} \rangle /N_{\text{ex}}$ for increasing $N_{\text{ex}}$, showing phase transition becomes sharp as $N_{\text{ex}} \rightarrow \infty$.}\label{fig:rescaledmeans}
\end{figure}


Figure \ref{fig:phaseboundaries} shows that the phase separation lines become sharper with larger $N_{ex}$, and a "true" phase transition emerges in infinite pumping rate limit. A similar conclusion can be drawn by plotting the rescaled stationary mean 
$\langle N\rangle / N_{ex}$, cf. Figure \ref{fig:rescaledmeans}. This can be intuitively understood by appealing to the effective potential  \eqref{eq.potential}. As $N_{ex}$ increases the potential barrier becomes larger and two stable phases emerge at the point where the local minima are equal. Indeed, Figure \ref{fig:potentials} shows the plot of the rescaled potential $U/N_{ex}$ as a function of the rescaled variable $x=n/N_{ex}$, which approaches the ($N_{ex}$ independent) limit 
$$
v(x) = - \int_0^x   \log \left ( \frac{\nu}{\nu+1}+ \frac{\sin^2(\alpha \sqrt{y} )}{(\nu+1)y } \right ) dy
$$
as it can be deduced from the formulas \eqref{eq:stationary} and \eqref{eq.potential}. Therefore, in the limit of large pumping rate we deal with a particle in a fixed potential $v(x)$ at inverse temperature $1/kT= N_{ex}$. At $\alpha\approx 1$ the dependence of the mean on $N_{\text{ex}}$ switches from constant to linear behaviour as the minimum of the potential $v(x)$ moves away from zero. When the two minima are at different heights, the lower one becomes the stable and other one is metastable. Communication between the phases becomes increasingly unlikely, with probability decreasing exponentially with $N_{ex}$. When the two minima are equal, we have two stable states, and the corresponding value of $\alpha$ is the phase transition point for the mean photon number.


More information about the dynamical phase transitions may be obtained from the rest of the spectrum of the semigroup generator, and for a more in-depth treatment of these numerical aspects, we refer the reader to \cite{vanHorssen}.
\section{Conclusions and outlook}

We have studied the counting process associated to the measurement of the outgoing atoms in the atom maser, and shown that this process satisfies the large deviations principle. In particular, this means that the cross-over behaviour observed in numerical simulations is not associated with the non-analyticity of the limiting log-moment generating function, as one would expect for a genuine phase transition. The rescaled counting process $\Lambda_{t}/N_{\text ex}$ does however appear to exhibit such a transition in the limit of infinite rate $N_{\text ex}$, as argued in the previous section using the potential model, and illustrated in Figures  \ref{fig:potentials}, \ref{fig:phaseboundaries}, and \ref{fig:rescaledmeans}. 
 
As a corollary, we have showed that the counting process satisfies the central limit theorem, which can be used to develop the statistical estimation theory of local asymptotic normality \cite{Guta2012,CatanaBoutenGuta,GutaKiukas}.


The model we have investigated has the property that the stationary state is diagonal in the Fock basis and all the jump operators leave the set of diagonal states invariant. The proof of the large deviation principle for the non-Markov counting process $\Lambda_t$ relies on the quantum semigroup's restriction to the diagonal algebra, which results in a classical birth-death semigroup (when proving the strict positivity and immediate compactness of $T^{(d)}_s$ and when applying Krein-Rutman theory).  We leave for a future investigation the study of the same problem in settings where no classical reduction is possible; an example would be the atom maser where the outgoing atoms are measured in a different basis than the standard one, thus breaking the invariance of the diagonal algebra.


The compactness of the Markov semigroup makes our model tractable as it becomes essentially finite dimensional, as the bath decay dominates the absorption due to the atom interaction. An interesting problem would be to explore more general classes of infinite dimensional systems (e.g. continuous variables or infinite spin chains) where a similar phenomenon holds. Another issue is the general relation between the "static" transitions which refer to non-analytic properties of the stationary state, and dynamic transitions which characterise properties of the measurement process. As shown in \cite{LesanovskyvanHorrsenGutaGarrahan} one can construct examples where the stationary state does not change while the system undergoes a dynamical phase transition.

Finally, it would be interesting to consider a more general large deviations setup  which takes into account the correlations between the detection events rather than looking at the total number of counts.  

\vspace{4mm}

\noindent

\emph{Acknowledgements.} M.v.H. and M.G. thank J. P. Garrahan and I. Lesanovsky for numerous fruitful discussions; their research was supported by EPSRC grants no. EP/J009776/1, EP/T022140/1, and Fellowship EP/E052290/1. R.C. and F.G aknowledge the support of the INDAM GNAMPA project 2020 ``Evoluzioni markoviane quantistiche'' and of the Italian Ministry of Education, University and Research (MIUR) for the the Dipartimenti di Eccellenza Program (2018–2022)—Dept. of Mathematics ``F. Casorati'', University of Pavia.

\emph{Conflict of interest.} The authors have no conflicts to disclose.

\emph{Data availability.} Data sharing is not applicable to this article as no new data were created or analyzed in this study.

\appendix
\section{Proofs of Propositions \ref{prop:def} and \ref{prop:defdef}}
\label{app:profs.prop1&2}

\renewcommand{\theequation}{\thesection \arabic{equation}}
\setcounter{equation}{0}
\noindent
\emph{Proof of Proposition \ref{prop:def}}

1. Since $G$ is negative by \eqref{eq.def.G}, it is the generator of a strongly continuous contraction semigroup on $\mathfrak{h}$. We have
$$
D(G) \subseteq D(L_i) \; \forall \, i=1,\dots,4
$$
and 
$$
\langle Gu,u\rangle + \langle u,Gu\rangle + \sum_{i=0}^4 \langle L_i u,L_i u\rangle=0 \quad \forall u \in D(G).
$$
Consequently $\mathcal{L}$ generates a minimal sub Markov quantum dynamical semigroup by \cite[Theorem 3.21]{Fagnola1999}. In particular, $\mathcal{L}^{(0)}$ generates the sub Markov quantum dynamical semigroup that we formally denoted by $(e^{t\mathcal{L}^{(0)}})_{t\ge 0}$.

In order to prove conservativity (which implies uniqueness by \cite[Corollary 3.23]{Fagnola1999}) we can use \cite[Corollary 3.41 p.73]{Fagnola1999} with $C=\Phi=-2G$ and $D=D(G)$; we just need to check condition (3.41) of \cite[Corollary 3.41 p.73]{Fagnola1999} which is equivalent to the existence of a constant $b$ such that for every $u \in D(G)$
\begin{equation}\label{eq.condition.G}
-4\langle Gu,Gu\rangle + \sum_{i=0}^4 \langle \sqrt{-2G} L_i u,\sqrt{-2G} L_i u\rangle 
\le b \|\sqrt{-2G}u\|^2
\end{equation}
Let $S$ be the shift operator, which acts as $S |e_n \rangle=|e_{n+1} \rangle$ and whose adjoint is $S^*|e_n \rangle=\delta_{n\ge 1}|e_{n-1} \rangle$; notice that $a=S^* \sqrt{N}$ and $a^*=\sqrt{N} S$. We recall some relations which are used in the following computation:
$$
aa^* =N+1,
\qquad
a^* a=N
$$
$$
Sf(N)| e_n \rangle = f(n) |e_{n+1}\rangle = f(N-1) S |e_n\rangle
$$
for any suitable function $f$ of the number operator. We have
\begin{eqnarray*}
-4G^2 - 2 \sum_{i=0}^4  L_i^*G L_i  
&=& -(N_{\rm ex} +\nu+ N (2 \nu +1))^2\\
&&+ {N_{\rm ex}}  \sin (\phi \sqrt{N+1}) S^* (N_{\rm ex} +\nu+ N (2 \nu +1)) S\sin (\phi \sqrt{N+1})\\
&&+ {N_{\rm ex}}  \cos^2 (\phi \sqrt{N+1}) (N_{\rm ex} +\nu+ N (2 \nu +1))\\
&&+ {(\nu+1)} a^* (N_{\rm ex} +\nu+ N (2 \nu +1)) a\\
&&+ {\nu} a (N_{\rm ex} +\nu+ N (2 \nu +1)) a^*\\
&=& -(N_{\rm ex} +\nu+ N (2 \nu +1))^2\\
&&+ {N_{\rm ex}}  \sin^2 (\phi \sqrt{N+1})  (N_{\rm ex} +\nu+ (N+1) (2 \nu +1)) \\
&&+ {N_{\rm ex}}  \cos^2 (\phi \sqrt{N+1}) (N_{\rm ex} +\nu+ N (2 \nu +1))\\
&&+ {(\nu+1)}N (N_{\rm ex} +\nu+ (N-1) (2 \nu +1)) \\
&&+ {\nu}  (N_{\rm ex} +\nu+ (N+1) (2 \nu +1)) (N+1)\\
&=& 0 \cdot N^2 -(2\nu +1) N + B \end{eqnarray*}
where $B$ is a bounded operator. Since $-2G= (N_{\rm ex} +\nu) + (2\nu+1) N$, we obtain \eqref{eq.condition.G}.

Equation (\ref{eq:integral1}) follows from \cite[Propositions 3.18]{Fagnola1999}; by the definition of minimal quantum dynamical semigroup given by Fagnola-Chebotarev we have that $\mathcal{T}$ is approximated in the pointwise $w^*$-topology by the following maps: for every $X \in B(\mathfrak{h})$, $u,v \in D(N)$
\be
\begin{split} \label{eq:iteration}
&\langle u,\mathcal{T}^{(0)}(t)(X) v \rangle= \langle e^{tG}u, X e^{tG}v\rangle,\\
&\langle u,\mathcal{T}^{(n+1)}(t)(X) v \rangle=\langle e^{tG}u, X e^{tG}v\rangle+ \sum_{i=1}^4 \int_0^t \langle L_i e^{sG}u, \mathcal{T}(t-s)^{(n)}(X)L_i e^{sG}v \rangle ds.
\end{split}
\ee
For every $n \geq 0$, equation (\ref{eq:iteration}) extends uniquely to a bounded bilinear form represented by a bounded operator $\mathcal{T}^{(n)}(t)(X)$. In our case for every $u \in \mathfrak{h}$ and $s>0$, $e^{sG}u \in D(N^\infty):=\bigcap_{n=1}^\infty D(N^n)$ and $L_i(D(N^\infty))\subset D(N^\infty)$, hence the expressions in equation (\ref{eq:iteration}) make sense for every $u,v \in \mathfrak{h}$ (the function which is inside the integral is well defined for every $u,v \in \mathfrak{h}$ unless when $s=0$, which is a set with zero Lebesgue measure). Moreover we can get by recursion the following explicit expression for $\langle u,\mathcal{T}^{(n+1)}(t)(X) v \rangle$, $u,v \in \mathfrak{h}$:
\be \label{eq:bilinear}
\begin{split}
\langle u,\mathcal{T}^{(n+1)}(t)&(X) v \rangle=\langle e^{tG}u, X e^{tG }v \rangle +\\
&\sum_{k=1}^n\sum_{i_1,\dots, i_n=1}^4 \int_{0\le t_1\le\cdots \le t_k\le t}f(t,u,v,X;t_1,i_1,\dots, t_k,i_k) dt_1\dots dt_k\\
\end{split}
\ee
where
\[\begin{split}
&f(t,u,v,X;t_1,i_1,\dots, t_k,i_k)=\\
&=\langle e^{G(t-t_k)}L_{i_k}\cdots  e^{G(t_2-t_1)}L_{i_1}e^{G t_1}u, X e^{G(t-t_k)}L_{i_k}\cdots  e^{G(t_2-t_1)}L_{i_1}e^{G  t_1}v \rangle.\\
\end{split}
\]
Since for every $t\geq 0$, $\mathcal{T}^{(n)}_t(X)$ converges to $\mathcal{T}_t(X)$ in the $w^*$-topology (monotonically for positive $X$), we conclude.
\\
2.-3. By restricting $\mathcal{T}_t$ to diagonal operators in the Fock basis, we obtain the semigroup of a birth-death process (see point (iii) 
in Section \ref{sec.main} for the details). We can then apply standard arguments (\cite[Section 5.4]{An91}) to see that $\rss$ is a faithful normal invariant state. This implies that the semigroup is positive recurrent (\cite{Frigerio1982}) and, consequently, the fixed points ${\cal F}({\cal T})$ are contained in the decoherence-free algebra ${\cal N}({\cal T})$. Moreover, thanks to \cite[Theorem 3.2]{Dhahri2010}, ${\cal N}({\cal T})$ is trivial (i.e. ${\cal N}({\cal T})=\mathbb{C} \mathbf{1}$). Then $\mathbf{1}$ is the unique harmonic projection and so $\rss$ is the unique invariant state; finally the semigroup is ergodic by \cite[Theorem 3.3]{Frigerio1982}.
\qed

\noindent
\emph{Proof of Proposition \ref{prop:defdef}}

1. and 2. We fix $s\in \mathbb{R}$ and we drop the index $s$ for the rest of the proof. We can proceed similarly as in the proof of Proposition \ref{prop:def}, and  use again \cite[Theorem 3.22]{Fagnola1999} choosing
\begin{eqnarray*}
L'_1&=&e^{s/2} L_1
\\
L_k'&=& L_k \qquad \mbox{for } k=2,3,4,
\\
 - 2G'&=& \sum_{k=1}^4 L_k^*L_k + (e^s-1)1_{s>0} \left ( 1+ \frac{1}{\nu} \right ) N_{ex} {\mathbf 1}= \\
&=&-2G+ (e^s-1)1_{s>0} \left ( 1+ \frac{1}{\nu} \right ) N_{{\rm ex}} {\mathbf 1}
\end{eqnarray*}
For the purposes of this proof, it would be enough to define $-2G^\prime$ as $-2G+(e^s-1)1_{s>0}N_{ex}{\mathbf 1}$; the extra factor $\left (1 + \frac{1}{\nu} \right)$ is required in order to keep the same notation in the proof of Lemma \ref{lem:l2ext}.

Notice that
\\
- $G'$ generates a strongly continuous contraction semigroup on $\mathfrak{h}$;
\\
- for all $u\in D(G')= D(N)$,
$$
\langle G'u, u\rangle + \langle u,G' u\rangle 
+ \sum_{k=1}^4 \langle L_k'u, L_k' u\rangle \le 0.
$$
\\
Then the operator $\mathcal{L}^\prime$, defined as
$$
\mathcal{L}^\prime(X) = \{G^\prime,X\} + \sum_{k=1}^4 {\mathcal J}_k^\prime(X)=\mathcal{L}^{(s)}(X)+ (1-e^s)1_{s>0} \left ( 1+ \frac{1}{\nu} \right ) N_{{\rm ex}} X
$$
generates a (sub-markovian) quantum dynamical semigroup $\mathcal{T}^\prime$ which is the minimal solution to
\begin{equation}\label{eq:integral}
\langle u,\mathcal{T}^\prime(t)(X) v \rangle=\langle e^{tG^\prime}u, X e^{tG^\prime}v\rangle+ \sum_{i=1}^4 \int_0^t \langle L^\prime_i e^{rG^\prime}u, \mathcal{T}^\prime_{t-r}(X)L^\prime_i e^{rG^\prime}v \rangle dr
\end{equation}
for every $u,v \in D(G^\prime)=D(G)$ and which is approximated by the sequence of maps
\[\begin{split}
&\langle u, \mathcal{T}^{(0)\prime }(t)(X) v \rangle= \langle e^{tG^\prime}u, X e^{tG^\prime}v\rangle,\\
&\langle u,\mathcal{T}^{(n+1) \prime}(t)(X) v \rangle=\langle e^{tG^\prime}u, X e^{tG^\prime}v\rangle+ \sum_{i=1}^4 \int_0^t \langle L^\prime_i e^{rG^\prime}u, \mathcal{T}^{(n) \prime}_{t-r}(X)L^\prime_i e^{rG^\prime}v \rangle dr.
\end{split}
\]
for $u,v \in \mathfrak{h}$. Notice that when $s\le 0$, $G=G^\prime$ hence we can take $\mathcal{T}_s:=\mathcal{T}^\prime$; while, when $s>0$, we can take $\mathcal{T}_s(t)=e^{t (e^s-1)\left ( 1+ \frac{1}{\nu} \right )N_{{\rm ex}}} \mathcal{T}^\prime(t)$, so again the conclusion immediately follows.\\
Let us prove the uniqueness of the solution to equation (\ref{eq:defsol}) and that its $w^*$-infinitesimal generator  $(\mathcal{L}_s,D(\mathcal{L}_s))$ is equal to $(\mathcal{L}+ (e^s-1){\mathcal J}_1,D(\mathcal{L}))$. Let us consider $\tilde{\mathcal{T}}$ the semigroup generated by $(\mathcal{L}_s+{\cal C},D(\mathcal{L}_s))$ (see Theorem \ref{thm:OB} in Appendix), where ${\cal C}$ is the bounded completely positive linear map defined as
\begin{equation}
{\cal C}(X) =(1-e^s){\mathcal J}_1(X)+1_{s>0}N_{{\rm ex}}(e^s-1)X,
\qquad X\in B(\mathfrak{h}).
\end{equation}
Because of Corollary \ref{coro:OB} in the Appendix, $\tilde{\mathcal{T}}$ is again a $w^*$-continuous quantum dynamical semigroup. Notice that Equation (\ref{eq:OB}) implies
\begin{equation} \label{eq:OB2}
\tilde{\mathcal{T}}(t)(X) = \tilde{\mathcal{T}}_s(t)(X) + \int_0^t
\tilde{\mathcal{T}}(t-r) {\cal C}\tilde{\mathcal{T}}(r)(X) dr
\end{equation}
and differentiating Equation (\ref{eq:OB2}) multiplied by $e^{t N_{{\rm ex}}(1-e^s) 1_{\{s>0\}}}$ we obtain that $e^{t N_{{\rm ex}}(1-e^s)1_{\{s>0\}}}\tilde{\mathcal{T}}$ solves for
\begin{equation} \label{eq:uniq}
\frac{d}{dt}\langle u, e^{t N_{{\rm ex}}(1-e^s)1_{\{s>0\}}}\tilde{\mathcal{T}}(t)(X) v \rangle =\mathcal{L}_*(\ket{v}\bra{u})(e^{t N_{{\rm ex}}(1-e^s)1_{\{s>0\}}}\tilde{\mathcal{T}}(t)(X)))
\end{equation}
for $u,v \in D(G)$. Since equation (\ref{eq:uniq}) admits a unique $w^*$-continuous positive solution (see \cite[Corollary 3.23]{Fagnola1999}), $e^{t N_{{\rm ex}}(1-e^s)1_{\{s>0\}}}\tilde{\mathcal{T}}$ must coincide with $\mathcal{T}$ and so do their $w^*$-infinitesimal generators, which are respectively $(\mathcal{L}_s+(1-e^s){\cal J}_1,D(\mathcal{L}_s))$ and $(\mathcal{L},D(\mathcal{L}))$. Hence we get that $(\mathcal{L}_s,D(\mathcal{L}_s))$ is equal to $(\mathcal{L}+ (e^s-1){\mathcal J}_1,D(\mathcal{L}))$ and the solution of Equation (\ref{eq:defsol}) is indeed unique.

3. By the previous arguments, we have, following once again the same line of the proof of Proposition \ref{prop:def}, the integral representation for the semigroup $\mathcal{T}$. When $s>0$, we just have to introduce the correcting multiplicative term.

\qed

\section{Additional results}
\label{appendix:other.restults}
\setcounter{equation}{0}

\begin{thm}[[Theorem 3.1.33, p. 191, \cite{Bratteli1979}] \label{thm:OB}
Let $\mathcal{S}$ be the generator of a $\sigma$-continuous semigroup $(\mathcal{P}(t))_{t \ge 0}$, with $\sigma$ equal to either the weak or the weak$^*$ topology.
If ${\cal C}$ is a bounded and $\sigma-\sigma$-continuous, then $({\cal S}+{\cal C})$ generates a $\sigma$-continuous semigroup ${\cal P}^{({\cal S}+{\cal C})}$ of bounded operators and for every $t \ge 0$, $X \in B(\mathfrak{h})$
\begin{equation} \label{eq:OB}
\begin{split}
&{\cal P}^{({\cal S}+{\cal C})}(t)(X) = {\cal P}(t)(x) + \\
&\sum_{k\ge 1} \int_{0\le t_1\le\cdots \le t_k \le t}
{\cal P}(t_1) {\cal C}{\cal P}(t_2-t_1){\cal C} \cdots {\cal P}(t_k-t_{k-1}){\cal C}{\cal P}(t-t_k)(X) dt_1\cdots dt_k\\
\end{split}
\end{equation}
The integrals define a series of bounded operators that converges in norm; the integrals are defined in the norm topology when $\sigma=w$ and in the weak$^*$ topology when $\sigma=w^*$.
\end{thm}

\begin{col} \label{coro:OB}
In the conditions of the previous theorem, if we additionally suppose that both ${\cal C}$ and the the maps ${\cal P}(t)$ are completely positive, then also the perturbed semigroup  ${\cal P}^{({\cal S}+{\cal C})}$ is completely positive.
\end{col}
\begin{proof}
Recall that the composition of two completely positive (c.p.) maps is still c.p. Then all the integrands in the integral form before are c.p..
Now remember another equivalent characterization of c.p.: $\Phi$ on $B(\mathfrak{h})$ is c.p. iff, for any $n\in \mathbb{N}$ and for any $X_1,...X_n$ and $Y_1,...Y_n$ bounded operators
$$
\sum_{i,j=1}^n X_i^*\Phi(Y_i^*Y_j)X_j \ge 0.
$$
Then it is immediate to see that 
\[\begin{split}
&\sum_{i,j=1}^n X_i^* {\cal P}^{({\cal S}+{\cal C})}(t)(Y_i^*Y_j)X_j = \sum_{i,j=1}^n X_i^* {\cal P}(t)(Y_i^*Y_j)X_j +\\
&\sum_{k\ge 1} \int_{0\le t_1\le\cdots \le t_k \le t}
\sum_{i,j=1}^n X_i^*{\cal P}(t_1) {\cal C}{\cal P}(t_2-t_1) \cdots {\cal C}{\cal P}(t-t_k)(Y_i^*Y_j)X_j dt_1\cdots dt_k \ge 0.\\
\end{split}
\]
\end{proof}

\begin{prop} \label{prop:core}
Let $B$ be a Banach space, $(S(t))_{t \ge 0}$ a strongly continuous semigroup with generator $(A,D(A))$. If ${\cal M}\subset B$ is such that
\begin{enumerate}
\item ${\cal M}$ is dense in $B$;
\item ${\cal M}$ is a set of analytic vectors for $A$, that is ${\cal M} \subset \bigcap_{n \ge 0}D(A^n)$ and for every $X \in {\cal M}$
\[\sum_{n=0}^{+\infty} \frac{z^n}{n!} \|A^n(X)\|
\]
has a positive radius of convergence;
\item $A({\cal M}) \subset {\cal M}$.
\end{enumerate}
Then ${\cal M}$ is a core for $A$.
\end{prop}
\begin{proof}
By the definition of core, we need to prove that ${\cal M}$ is dense in $D(A)$ with respect to the graph norm $\lVert X \rVert_{A}:=\lVert X \rVert + \lVert {A}(X) \rVert$. Fix $X \in D(A)$, $\epsilon >0$. We shall consider successive approximations of $X$.
\paragraph{Step 1}
\[\lim_{t \rightarrow 0} \Big\lVert X-\frac{1}{t} \int_0^t S(u)(X) du \Big\rVert_{A} =0.
\]
Notice that, due to the strong continuity of $S$,
\[\lim_{t \rightarrow 0} \Big\lVert X-\frac{1}{t} \int_0^t S(u)(X) du \Big\rVert =0
\]
is trivial and that 
\[\lim_{t \rightarrow 0} \Big\lVert {A}(X)-\frac{1}{t} {A}\int_0^t S(u)(X) du \Big\rVert=\lim_{t \rightarrow 0} \Big\lVert {A}(X)-\frac{S(t)(X)-X}{t} \Big\rVert=0
\]
follows from the definition of infinitesimal generator. Therefore there exists $\tilde{t}>0$ such that
\[\Big\lVert X-\frac{1}{\tilde{t}} \int_0^{\tilde{t}} S(u)(X) du \Big\rVert_{A} \leq \epsilon.
\] 
\paragraph{Step 2} There exists $Y \in {\cal M}$ such that
\[\Big\lVert \frac{1}{\tilde{t}} \int_0^{\tilde{t}}S(u)(X)du -  \frac{1}{\tilde{t}} \int_0^{\tilde{t}} S(u)(Y)du \Big\rVert_{A} <\epsilon.
\]
Indeed, ${\cal M}$ is dense in $B$, hence we can find $Y \in {\cal M}$ such that $\Big\lVert X-Y \Big\rVert < \epsilon \tilde{t}/4$; we have
\[\Big\lVert \frac{1}{\tilde{t}} \int_0^{\tilde{t}} S(u)(X)du -  \frac{1}{\tilde{t}} \int_0^{\tilde{t}} S(u)(Y)du \Big\rVert \leq \frac{1}{\tilde{t}} \int_0^{\tilde{t}} \Big\lVert S(u)(X-Y) \Big\rVert du \leq \epsilon/2
\]
and
\[\Big\lVert \frac{1}{\tilde{t}}{A} \int_0^{\tilde{t}} S(u)(X)du -  \frac{1}{\tilde{t}} {A}\int_0^{\tilde{t}} S(u)(Y)du \Big\rVert=\frac{1}{\tilde{t}} \Big\lVert S_{\tilde{t}}(X-Y)+(X-Y) \Big\rVert \leq \epsilon/2
\]
and, summing up, we conclude.
\paragraph{Step 3} If we assume that $\overline{t}$ is small enough, there exists $N>0$ such that
\[\Big\lVert \frac{1}{\tilde{t}}\int_0^{\tilde{t}} S(u)(Y)du - \frac{1}{\tilde{t}}\int_0^{\tilde{t}}\sum_{k=0}^N \frac{u^k {A}^k(Y)}{k!} du\Big\rVert <\epsilon/2
\]
and, since ${A}(Y) \in {\cal M}$ too,
\[\Big\lVert  \frac{1}{\tilde{t}}\int_0^{\tilde{t}} S(u)({A}(Y))du - \frac{1}{\tilde{t}}\int_0^{\tilde{t}} \sum_{k=0}^N \frac{u^k {A}^k({A}(y))}{k!}du \Big\rVert <\epsilon/2.
\]
Hence
\[\Big\lVert \frac{1}{\tilde{t}}\int_0^{\tilde{t}} S(u)(Y)du - \frac{1}{\tilde{t}}\int_0^{\tilde{t}}\sum_{k=0}^N \frac{u^k {A}^k(Y)}{k!} du\Big\rVert_{A} <\epsilon/2
\]
and notice that
$$
\frac{1}{\tilde{t}}\int_0^{\tilde{t}}\sum_{k=0}^N \frac{u^k A^k(Y)}{k!} du= \sum_{k=0}^N\left (\frac{1}{\tilde{t}}\int_0^{\tilde{t}} \frac{u^k }{k!}du\right)A^k(Y)\in {\cal M}.
$$
Hence we have shown that for every $X \in D({A})$ and for every $\epsilon >0$, there exists an element in ${\cal M}$ which is far from $X$ less than $4\epsilon$ with respect to $\lVert \cdot \rVert_{A}$ and we are done.
\end{proof}

\bibliography{bibliography}
\bibliographystyle{utphys}

\end{document}